\documentclass[sigconf,nonacm]{aamas}
\usepackage{balance}
\usepackage{amsmath}
\usepackage{booktabs}
\usepackage{multirow}
\usepackage{tabularx}
\usepackage{array}
\usepackage{placeins}
\usepackage{algorithm}
\usepackage[noend]{algpseudocode}
\algrenewcommand\algorithmicrequire{\textbf{Input:}}
\algrenewcommand\algorithmicensure{\textbf{Output:}}
\setcopyright{none}
\renewcommand\footnotetextcopyrightpermission[1]{}
\makeatletter
\AtBeginDocument{%
  \fancypagestyle{firstpagestyle}{\fancyhf{}\renewcommand{\headrulewidth}{\z@}\renewcommand{\footrulewidth}{\z@}%
    \fancyfoot[C]{\footnotesize\thepage}}%
  \fancypagestyle{standardpagestyle}{\fancyhf{}\renewcommand{\headrulewidth}{\z@}\renewcommand{\footrulewidth}{\z@}%
    \fancyhead[L]{\@headfootfont\shorttitle}\fancyhead[R]{\@headfootfont\shortauthors}%
    \fancyfoot[C]{\footnotesize\thepage}}%
  \pagestyle{standardpagestyle}}
\makeatother
\title[RAVEN: Receiver-Conditioned Action-Value Encoding]{RAVEN: Receiver-Conditioned Action-Value Encoding for Finite-Alphabet Multi-Agent Communication}
\author{Shuwei Sun}
\orcid{0009-0009-4809-8557}
\affiliation{\institution{Xi'an Jiaotong University}\city{Xi'an}\country{China}}
\author{Chenxi Wang}
\orcid{0009-0006-2544-5526}
\affiliation{\institution{Xi'an Jiaotong University}\city{Xi'an}\country{China}}
\author{Jian Huang}
\orcid{0009-0007-8958-5736}
\affiliation{\institution{Xi'an Jiaotong University}\city{Xi'an}\country{China}}
\author{Weiyun Ru}
\orcid{0009-0007-7650-8555}
\affiliation{\institution{Xi'an Jiaotong University}\city{Xi'an}\country{China}}
\author{Hui Cao}
\orcid{0000-0002-4985-0028}
\affiliation{\institution{Xi'an Jiaotong University}\city{Xi'an}\country{China}}
\renewcommand{\shortauthors}{Sun et al.}

\newcommand{\sg}{\operatorname{sg}}
\newcommand{\dep}{\mathrm{dep}}
\newcommand{\refb}{\mathrm{ref}}
\newcommand{\LN}{\mathrm{LN}}
\newcommand{\E}{\mathbb{E}}
\newcommand{\Lcv}{\mathcal{L}_{\mathrm{CV}}}
\newcommand{\Ltd}{\mathcal{L}_{\mathrm{TD}}}
\newcommand{\off}{RAVEN\textsubscript{off}}
\newcommand{\on}{RAVEN\textsubscript{on}}

\begin{abstract}
A message drawn from a small alphabet helps a teammate only if it keeps the distinctions that change that teammate's next decision. We show that scoring messages by action values averaged over the receiver's situation can erase exactly these distinctions, and we propose RAVEN (Receiver-conditioned Action-Value ENcoding), which trains a four-symbol, one-step-delayed channel to preserve each receiver's centered action-value profile within the receiver's own context. The sender never needs to know that context: the receiver decodes every symbol with its private information. We give two estimators of this target. With a teacher, offline RAVEN selects the codebook that exactly minimizes an empirical conditional distortion and distills it into a frozen sender; we bound the resulting codebook-selection error and one-step decision loss. Without a teacher, online RAVEN aligns, inside a QMIX learner, the deployed symbol pathway with a training-only continuous reference that shares its routing. Against five recent communication methods on eight navigation settings, offline RAVEN attains the highest return in seven, and removing receiver conditioning forfeits 83\% of its communication gain. Online RAVEN raises predator--prey capture success from 53.2\% to 96.0\% over the same QMIX backbone without communication, and on SMAC and MPE it attains the best mean normalized score of 14 methods, including methods that exchange kilobit messages. Every RAVEN message costs 2~bits, 12--1{,}024$\times$ fewer than those of NDQ, CACOM and ExpoComm on navigation.
\end{abstract}

\begin{document}
\maketitle

\section{Introduction}
Agents that cooperate under partial observability often share a channel that carries only a few bits per step: a radio link, an acoustic modem, or a protocol budget in a large team. With an alphabet of four symbols, a sender cannot pass on what it sees; it has to decide which distinctions a symbol keeps. The natural answer is to keep the distinctions that change what the \emph{receiver} should do. We show that this answer depends on \emph{where} those distinctions are measured, and we build a communication method on the resulting principle.

Learned communication has progressed from continuous message passing~\cite{r1,r2} to attention, gating and scheduling~\cite{r22,r23,r24,r6,r7} and to succinct or bottlenecked messages~\cite{r9,r43,r44,r16,r8}. Decision-oriented methods go further and compress observations by their effect on value, for example by clustering value representations with a return-gap guarantee~\cite{r3} or by aggregating observations according to inferred optimal actions~\cite{r4}. Receiver-aware methods tailor messages to their recipients by modeling teammates~\cite{r45} or by requesting their context before replying~\cite{r5}. These lines establish that messages should serve decisions; the question we address is how a sender that sees only its own input, and speaks through four symbols with one step of delay, should decide which of its situations deserve different symbols.

Consider two sender types and two receiver contexts (Figure~\ref{fig:example}). In context $C_1$ the receiver's centered values over its two actions are $(1,-1)$ under type $T_1$ and $(-1,1)$ under $T_2$; in $C_2$ they are reversed. Averaged over contexts, both types have the profile $(0,0)$: an objective that compares value profiles after averaging sees no reason to separate them, and the message becomes useless. Compared within each context, the two types disagree completely, and separating them lets the receiver act optimally in both. Conditioning on everything the receiver observes avoids the cancellation but leaves too few samples per context. RAVEN takes the middle road. It compares centered receiver action values within a small number of \emph{receiver conditions}, formed from the receiver features to which its action gaps are most sensitive, and it never transmits those conditions: the receiver interprets each symbol with its own information, much as a decoder uses side information in Wyner--Ziv coding~\cite{r17}. One symbol can therefore mean ``go left'' in one context and ``go right'' in another at no cost in bandwidth.

Two failure modes must be kept apart here. Constructive objectives that score candidate symbols by value profiles averaged over the receiver's situation---the value-clustering constructions of RGMComm and ABSA~\cite{r3,r4}---face the cancellation head-on, because the averaging is part of their objective. End-to-end methods such as DIAL and SLIM~\cite{r2,r8} instead backpropagate the receiver's loss, which is computed from the receiver's own observation, into the sender; their messages are receiver-conditioned by construction, and what limits them is the biased gradient of the discrete channel, not cancellation. RAVEN addresses each in turn: the conditional target removes the first, and exact construction with a frozen sender avoids the second. Section~\ref{sec:why} isolates the two effects with separate ablations.

\begin{figure}[t]
\centering
\includegraphics[width=\columnwidth]{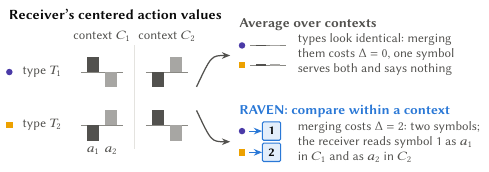}
\caption{Why the receiver's context matters. Two sender types induce opposite action gaps in two receiver contexts. After averaging over contexts both profiles are zero and merging the types costs nothing ($\Delta=0$, Eq.~\eqref{eq:codebook}); compared within contexts they must be kept apart ($\Delta=2$), and the receiver decodes the same symbol differently in each context.}
\label{fig:example}
\Description{Left: a two-by-two grid of bar pairs showing centered values of two receiver actions for two sender types in two contexts; the signs flip between contexts. Top right: after averaging, both types have flat profiles and would share one symbol. Bottom right: RAVEN assigns them different symbols, and the receiver maps symbol 1 to different actions in the two contexts.}
\end{figure}

We turn this principle into two estimators that share one execution interface (Figure~\ref{fig:overview}). \emph{Offline RAVEN} uses a centralized teacher to build a table of receiver-conditioned action values, selects the optimal four-symbol codebook for this table exactly, by enumerating all 611{,}501 groupings of twelve sender types, distills it into a sender network that is then frozen, and trains receivers on the frozen symbols. \emph{Online RAVEN} removes the teacher: inside a QMIX learner~\cite{r14}, a training-only reference branch replaces each decoded symbol with the sender's continuous state under exactly the same routing, and an alignment loss pulls the deployed branch's action-value profile toward the reference profile.

Our contributions are:
\begin{itemize}
\item \textbf{A receiver-conditioned target for finite-alphabet messages.} We define the target of a symbol as the receiver's centered action-value profile within its own condition, decompose its reconstruction risk into an abstraction and a compression term, and bound both the codebook-selection error and the one-step decision loss (Section~\ref{sec:analysis}).
\item \textbf{Offline RAVEN}, an exact conditional codebook with a frozen sender. Against SLIM, NDQ, CACOM, ExpoComm and MACC on eight navigation settings, it has the highest return in seven, wins 38 of 40 seed-paired comparisons, and needs 82--98\% fewer execution FLOPs than NDQ, CACOM and ExpoComm.
\item \textbf{Online RAVEN}, a teacher-free paired-reference estimator. With the backbone held fixed, its 2-bit channel raises predator--prey capture success by 42.8 percentage points, and on SMAC and MPE it has the best mean normalized score, ahead of methods that exchange 2--4~kbit messages.
\item \textbf{Evidence for the mechanism.} Conditioning on the receiver accounts for 83\% of the communication gain, and the conditional risk that RAVEN minimizes, measured on held-out data, orders the variants by return (within-seed Spearman $\rho=0.76$).
\end{itemize}

\section{Related Work}
\textbf{Learned communication under bandwidth limits.} CommNet and DIAL learn continuous and discretized messages end to end~\cite{r1,r2}; ATOC, TarMAC and IC3Net learn groups, targeted attention and gates~\cite{r22,r23,r24}; ExpoComm and SOPS scale communication with exponential and sparse topologies~\cite{r6,r7}. NDQ, IMAC and TMC make messages succinct with information-theoretic regularizers and temporal message control~\cite{r9,r43,r44}, VQ-VIB trades utility against vocabulary complexity~\cite{r16}, and SLIM decouples policy and message dimensions under bandwidth constraints~\cite{r8}. MASIA, GLC and LMAC ground messages by reconstruction, language-model supervision or state recovery~\cite{r27,r28,r29}. These methods learn \emph{what} to send by end-to-end gradients or reconstruction; RAVEN instead specifies the target of each symbol, the receiver's conditional action-value profile, and optimizes it exactly when a teacher is available.

\textbf{Decision-oriented and receiver-aware messages.} The RGMComm method bounds the return gap of discrete messages obtained by clustering value representations~\cite{r3}, and ABSA aggregates observations by their inferred optimal actions for a centralized controller~\cite{r4}. MAIC generates incentive messages from models of teammates~\cite{r45}; CACOM broadcasts a context request and returns personalized, gated replies~\cite{r5}; SeqComm-DFL scores sequential messages by the improvement of a receiver's best value~\cite{r21}; MACC assigns credit to delayed discrete messages with a communication critic~\cite{r20}. RAVEN shares the decision-oriented view but differs in two respects: its distortion is measured within receiver conditions before averaging, which prevents opposite action gaps from canceling, and it needs no request round, since one broadcast symbol is decoded by each receiver with its own information---the side-information principle of Wyner--Ziv coding~\cite{r17}.

\textbf{Task-oriented communication over physical channels.} The wireless and semantic-communication communities study the same joint design of learning and communication: Tung et al.~\cite{r51} treat a noisy channel as part of the multi-agent environment and learn coding and control jointly, and task-based information compression~\cite{r52} formulates observation compression under rate constraints as a rate--distortion problem whose distortion is measured in task return. RAVEN shares the task-oriented objective but assumes an error-free channel with a four-symbol alphabet; the scarce resource is not the signal-to-noise ratio but the distinctions that four symbols can keep.

\textbf{Value factorization and supervision.} We build on centralized training with decentralized execution~\cite{r10,r11,r12,r15} and on monotonic value factorization~\cite{r13,r14,r40}, which RAVEN uses to express each receiver action in team value during training. The offline estimator relates to policy distillation~\cite{r18} and to vector quantization~\cite{r25,r26,r34}; unlike policy distillation, RAVEN distills a codebook whose meaning is fixed by construction, and freezes it so that the receivers can rely on it.

\section{RAVEN}\label{sec:method}
\begin{figure*}[t]
\centering
\includegraphics[width=\textwidth]{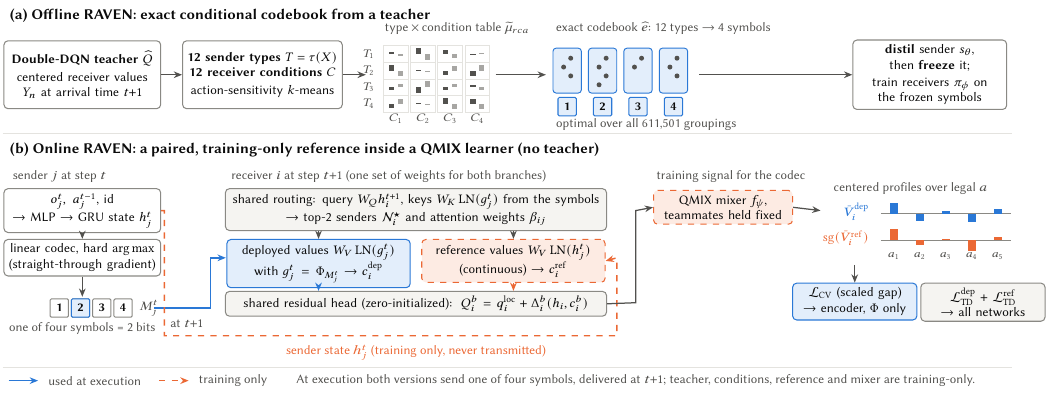}
\caption{RAVEN under a four-symbol, one-step-delayed interface. (a) Offline RAVEN builds a table of receiver-conditioned centered values from a teacher, selects the exactly optimal codebook, distills and freezes the sender, and trains the receivers on its symbols (table and grouping shown schematically). (b) Online RAVEN evaluates each receiver twice with one set of weights and one routing decision: once with the decoded symbols (deployed) and once with the senders' continuous states (reference, training only). The alignment loss $\Lcv$ between the two centered profiles updates only the encoder and the symbol embedding $\Phi$; TD losses train all networks.}
\label{fig:overview}
\Description{Two lanes. The offline lane goes from a Double-DQN teacher through sender types and receiver conditions to a type-by-condition table, an exact codebook mapping twelve types to four symbols, and a frozen sender with trained receivers. The online lane shows a recurrent sender producing a hard symbol delivered at the next step, a receiver whose shared routing feeds a deployed branch using symbol embeddings and a training-only reference branch using the sender's continuous state, a shared residual head, a QMIX mixer with fixed teammates, centered action profiles of both branches, and the two losses.}
\end{figure*}

\subsection{Interface and Target}\label{sec:interface}
We consider $N$ cooperative agents with shared reward $r_t$ and discount $\gamma$ in a decentralized, partially observable environment~\cite{r10}. At step $t$, each active sender $j$ maps its permitted local input $X_j^t$ (its own observation and history; no global state and no other agent's private state) to one of $K=4$ symbols,
\begin{equation}\label{eq:sender}
M_j^t=s_{\theta,j}(X_j^t)\in\{1,\ldots,K\},
\end{equation}
which the task's communication edges deliver at $t+1$. Receiver $i$ then acts on its local information $Z_i^{t+1}$ and its mailbox $\mathbf M_i^{t+1}$ of arrived symbols. Each message costs $\log_2K=2$ bits, and the symbols are the only quantities exchanged at execution.

A receiver's decision depends on the differences between its action values, not on their level. For a legal action set $\mathcal A_i$, let $P=I-\mathbf 1\mathbf 1^\top/|\mathcal A_i|$ remove the mean of a value vector; then
\begin{equation}\label{eq:center}
\|P(u-v)\|_2^2=\frac{1}{|\mathcal A_i|}\sum_{a<b}\big[(u_a-u_b)-(v_a-v_b)\big]^2 ,
\end{equation}
so the error between centered profiles is exactly the error in all pairwise action gaps. RAVEN asks each symbol to preserve the receiver's centered action-value profile \emph{within the receiver's condition}, a finite summary of the receiver's context that is used during training only.

\subsection{Offline RAVEN: Exact Conditional Codebook}\label{sec:offline}
\textbf{Targets.} A Double-DQN teacher $\widehat Q$~\cite{r37} is trained with centralized inputs. Each construction row $n\in\mathcal D$ pairs a sender input $X_n$ at time $t$ with the receiver's features $V_n\subseteq Z_n$ and the teacher context $\xi_n$ at the arrival time $t+1$, so the target describes the decision the receiver faces when the symbol arrives. With partner actions drawn from a fixed reference $\rho$ (uniform in navigation), the receiver's centered, scaled profile is
\begin{equation}\label{eq:target}
\begin{aligned}
u_{na}&=\E_{\mathbf a_{-i}\sim\rho}\,\widehat Q(\xi_n,a,\mathbf a_{-i}),\\
Y_{na}&=\frac{1}{s_Q}\Big(u_{na}-\frac{1}{|\mathcal A_n|}\sum_{b\in\mathcal A_n}u_{nb}\Big),
\end{aligned}
\end{equation}
where $s_Q$ is the standard deviation of teacher values over $\mathcal D$.

\textbf{Sender types and receiver conditions.} Sender inputs are clustered into $n_T=12$ types $T=\tau(X)$ (speaker--listener tasks use the three target colors). Receiver conditions are obtained by clustering in a metric that stretches the receiver features to which its action gaps are sensitive. With standardized features $\bar V\in\mathbb R^F$,
\begin{equation}\label{eq:metric}
\begin{aligned}
\zeta_\ell&=\frac{1}{|\mathcal D|}\sum_{n}\frac{1}{|\mathcal A_n|}\sum_{a\in\mathcal A_n}\Big(\frac{\partial Y_{na}}{\partial\bar V_{n\ell}}\Big)^{2},\\
C&=\operatorname{kmeans}\big(\Omega^{1/2}\bar V\big),\qquad
\Omega=\operatorname{diag}\Big(\frac{F\,\zeta_\ell}{\sum_k\zeta_k}\Big)+0.05I,
\end{aligned}
\end{equation}
with $n_C=12$ conditions, k-means++ seeding and Lloyd iterations~\cite{r31,r32}.

\textbf{Conditional table.} For type $r$, condition $c$ and action $a$, let $N_{rca}$ and $S_{rca}$ be the mass and the target sum of the rows with $T_n=r$, $C_n=c$ and $a$ legal, each row spreading unit weight over its legal actions. Sparse cells are shrunk toward the type mean $\bar Y_{ra}$:
\begin{equation}\label{eq:shrink}
\widetilde\mu_{rca}=\frac{S_{rca}+\lambda_{ra}\bar Y_{ra}}{N_{rca}+\lambda_{ra}},\qquad
\alpha_{rca}=\frac{N_{rca}}{|\mathcal D|},
\end{equation}
with pseudo-count $\lambda_{ra}=32N_{ra}/n_r$, where $N_{ra}=\sum_cN_{rca}$ and $n_r$ is the number of rows of type $r$.

\textbf{Exact codebook.} A codebook $e$ maps the types to at most $K$ symbols. Its distortion is the mass-weighted error between each cell and the centroid of its symbol \emph{in the same condition},
\begin{equation}\label{eq:codebook}
\widehat D(e)=\sum_{r,c,a}\alpha_{rca}\big(\widetilde\mu_{rca}-\bar\mu^{e}_{e(r)ca}\big)^2,\quad
\bar\mu^{e}_{mca}=\frac{\sum_{r:e(r)=m}\alpha_{rca}\widetilde\mu_{rca}}{\sum_{r:e(r)=m}\alpha_{rca}} .
\end{equation}
Merging two groups with masses $v_1,v_2$ and means $\mu_1,\mu_2$ raises $\widehat D$ by $\Delta=\sum_{c,a}\frac{v_{1ca}v_{2ca}}{v_{1ca}+v_{2ca}}(\mu_{1ca}-\mu_{2ca})^2$~\cite{r35}, which vanishes only if the groups agree in every condition where both occur; opposite gaps in different contexts therefore never cancel (Figure~\ref{fig:example}), whereas groups that never occur in the same condition can share a symbol for free. For twelve types we enumerate all $S(12,4)=611{,}501$ partitions into four groups; because refinement never increases distortion, this returns the exact minimizer $\widehat e$ over all codebooks with at most four symbols. With several recipients, one shared codebook minimizes the sum of their distortions.

\textbf{Sender and receivers.} The sender network $p_\theta(m\mid X)$ is fitted to the labels $\widehat e(\tau(X_n))$ by cross-entropy, $s_\theta(X)=\arg\max_m p_\theta(m\mid X)$, and then frozen together with every feature that produces symbols, so the meaning of each symbol cannot drift. Receivers keep their full local input $Z_i$, not the coarse condition, and are trained on the frozen senders' actual symbols to maximize the teacher's centered joint value under their product policy,
\begin{equation}\label{eq:receiver}
\mathcal L_{\rm rec}=-\E_n\Big[\sum_{\mathbf a}\prod_i\pi_{\phi,i}(a_i\mid Z_{i},\mathbf M_{i})\,\widetilde Q_n(\mathbf a)\Big],
\end{equation}
where $\widetilde Q_n$ is the teacher value centered over legal joint actions and scaled by $s_Q$. At execution, senders emit their $\arg\max$ symbol and receivers act greedily; neither types nor conditions are ever computed.

\subsection{Online RAVEN: Paired Reference}\label{sec:online}
When no teacher is available or joint actions cannot be enumerated, RAVEN estimates the receiver-conditioned target inside the learner (Figure~\ref{fig:overview}b). The idea is to ask, during training only, what the receiver would do if it could read the sender's continuous state instead of a 2-bit symbol, and to teach the four symbols to reproduce that decision profile: a training-only \emph{reference branch} answers the question, and the alignment loss $\Lcv$ does the teaching. Agents share a recurrent network (a 64-unit layer and a 64-unit GRU~\cite{r33}) that maps observation, previous action and identity to a state $h_i^t$ and local utilities $q_i^{{\rm loc},t}$; a QMIX mixer $f_\psi$~\cite{r14} combines utilities during training.

\textbf{Codec.} A linear encoder produces four logits $l_j^t$. The transmitted symbol is the $\arg\max$; the backward pass uses the straight-through estimator~\cite{r36},
\begin{equation}\label{eq:st}
\widetilde m_j^t=\operatorname{onehot}(\arg\max l_j^t)+p_j^t-\sg(p_j^t),\quad p_j^t=\operatorname{softmax}(l_j^t),
\end{equation}
and a shared $4\times64$ embedding $\Phi$ decodes $g_j^t=\Phi^\top\widetilde m_j^t$, the embedding of the received symbol.

\textbf{One routing, two branches.} At $t+1$, receiver $i$ scores each permitted sender by the scaled dot product of its private query $W_Qh_i^{t+1}$ and the symbol key $W_K\LN(g_j^t)$, keeps the two highest-scoring senders $\mathcal N_i^\star$, and normalizes their scores into weights $\beta_{ij}$. RAVEN evaluates two branches under this single routing decision:
\begin{equation}\label{eq:contexts}
c_i^{\dep}=\sum_{j\in\mathcal N_i^\star}\beta_{ij}W_V\LN(g_j^t),\quad
c_i^{\refb}=\sum_{j\in\mathcal N_i^\star}\beta_{ij}W_V\LN(h_j^t).
\end{equation}
The deployed branch aggregates what was transmitted; the reference branch aggregates the continuous states that the symbols had to summarize. Because query, keys, selected senders and weights are shared, the two branches differ only in the content of the messages. A residual head $R_\vartheta$ with a zero-initialized output layer turns each context into an action correction,
\begin{equation}\label{eq:utilities}
\begin{aligned}
Q_i^{b}(a)&=q_i^{\rm loc}(a)+\mathbf 1\{\mathbf M_i\neq\varnothing\}\,\Delta_i^{b}(a),\\
\Delta_i^{b}&=R_\vartheta\big([h_i,c_i^{b},h_i\odot c_i^{b}]\big),\qquad b\in\{\dep,\refb\},
\end{aligned}
\end{equation}
so training starts from the local policy and the channel gains influence only as it becomes useful.

\textbf{Paired profiles.} For a probe $(t,i)$ sampled from replay, we hold the teammates' utilities $u_k$ at their reference values on the replayed actions, vary receiver $i$'s action over its legal set, and pass both branches through the same mixer and state:
\begin{equation}\label{eq:profiles}
V_i^{b}(a)=f_\psi\big(u_1,\ldots,Q_i^{b}(a),\ldots,u_N;s\big),\qquad \bar V_i^{b}=PV_i^{b}.
\end{equation}
The mixed profiles express every receiver action in team value, and centering removes their level. The alignment loss moves the deployed profile toward the reference profile,
\begin{equation}\label{eq:cv}
\begin{aligned}
\Lcv&=\frac{1}{n_{\mathcal B}}\sum_{(t,i)\in\mathcal B}\ \sum_{a\in\mathcal A_i^t}\frac{\big[\bar V_i^{\dep}(a)-\sg\big(\bar V_i^{\refb}(a)\big)\big]^2}{\sg(\sigma_i^2)},\\
\sigma_i^2&=\max\Big\{1,\ \frac{1}{|\mathcal A_i^t|}\sum_{a}\bar V_i^{\refb}(a)^2\Big\},
\end{aligned}
\end{equation}
where $\mathcal B$ contains at most eight probe times per update and $n_{\mathcal B}$ counts its receiver--action pairs. The gradient of $\Lcv$ reaches only the encoder and the embedding $\Phi$: states, teammate utilities, reference profiles and scales are detached, and the receiver and mixer weights are held fixed for this term. Both branches are trained by Double-Q TD losses on the team value~\cite{r37,r14},
\begin{equation}\label{eq:total}
\mathcal L=\tfrac12\big(\Ltd^{\dep}+\Ltd^{\refb}\big)+0.1\,\Lcv .
\end{equation}
At execution, only the recurrent network, codec, receiver and local head remain, and agents act greedily on their deployed utilities; the reference branch, the mixer and the global state are used in training only, and no continuous state is ever transmitted.

\begin{algorithm}[t]
\caption{RAVEN training}
\label{alg:raven}
\small
\textit{Offline RAVEN} (teacher $\widehat Q$, construction data $\mathcal D$, $K=4$)
\begin{algorithmic}[1]
\State compute centered targets $Y_n$ at the arrival time (Eq.~\ref{eq:target})
\State cluster sender types $T=\tau(X)$ and receiver conditions $C$ (Eq.~\ref{eq:metric})
\State build the shrunk table $\widetilde\mu_{rca}$ with masses $\alpha_{rca}$ (Eq.~\ref{eq:shrink})
\State $\widehat e\gets\arg\min_e\widehat D(e)$ by enumerating all groupings (Eq.~\ref{eq:codebook})
\State fit the sender $s_\theta$ to the labels $\widehat e(\tau(X_n))$, then \textbf{freeze} it
\State train receivers $\pi_\phi$ on the frozen symbols (Eq.~\ref{eq:receiver})
\end{algorithmic}
\smallskip
\textit{Online RAVEN} (one update, after each collected episode)
\begin{algorithmic}[1]
\State act $\epsilon$-greedily with hard symbols; store the episode in replay
\State sample 32 episodes; unroll states $h$ and symbols $M$
\For{$b\in\{\dep,\refb\}$} \Comment{one routing, two branches}
\State contexts $c^{b}$ (Eq.~\ref{eq:contexts}), utilities $Q^{b}$ (Eq.~\ref{eq:utilities}), TD loss $\Ltd^{b}$
\EndFor
\State for $\le8$ probe times: paired profiles (Eq.~\ref{eq:profiles}) and $\Lcv$ (Eq.~\ref{eq:cv})
\State take one gradient step on $\mathcal L$ (Eq.~\ref{eq:total}); $\nabla\Lcv$ reaches the codec only
\end{algorithmic}
\end{algorithm}

\subsection{Analysis}\label{sec:analysis}
We state the guarantees of the offline estimator for a fixed teacher, partner reference and input distribution, a fixed legal set with $\|v\|_{\mathcal A}^2=|\mathcal A|^{-1}\sum_a v_a^2$, and the ideal symbol $M=e(T)$. Proofs are in Appendix~\ref{app:proofs}.

\begin{proposition}[Decomposition and codebook selection]\label{prop:selection}
Let $\mathcal R_C(e)=\E\|Y-\E[Y\mid M,C]\|_{\mathcal A}^2$, $\mathcal R_Z(e)=\E\|Y-\E[Y\mid M,Z]\|_{\mathcal A}^2$ and $\mathcal G_C(e)=\E\|\E[Y\mid M,Z]-\E[Y\mid M,C]\|_{\mathcal A}^2$. Since $C$ is a function of $Z$,
\begin{equation}\label{eq:decomposition}
\mathcal R_C(e)=\mathcal R_Z(e)+\mathcal G_C(e)=\mathcal E_{\rm abs}+\mathcal D(e),
\end{equation}
with abstraction error $\mathcal E_{\rm abs}=\E\|Y-\E[Y\mid T,C]\|_{\mathcal A}^2$, independent of $e$, and compression error $\mathcal D(e)=\E\|\E[Y\mid T,C]-\E[Y\mid M,C]\|_{\mathcal A}^2$. If $\sup_e|\widehat D(e)-\mathcal D(e)|\le\epsilon_D$ and $\widehat e$ minimizes $\widehat D$ up to $\epsilon_{\rm opt}$, then
\begin{equation}\label{eq:selection}
\mathcal R_Z(\widehat e)-\inf_e\mathcal R_Z(e)\le\operatorname{osc}_e\mathcal G_C(e)+2\epsilon_D+\epsilon_{\rm opt}.
\end{equation}
\end{proposition}

\begin{proposition}[One-step decision loss]\label{prop:decision}
Let $U$ be the receiver's true payoffs under the partner reference with $|U_a|\le B$, $\epsilon_Q=\E\|U-u\|_\infty$ for the teacher vector $u$ of Eq.~\eqref{eq:target}, $q=\Pr[s_\theta(X)\neq\widehat e(T)]$, and let the receiver lose at most $\epsilon_{\rm rec}$ in true payoff relative to the decoder $\arg\max_a\E[u_a\mid M,Z]$. Then its expected payoff $J_U$ satisfies
\begin{equation}\label{eq:decision}
\E\max_aU_a-J_U\le2\epsilon_Q+s_Q\sqrt{|\mathcal A|\,\mathcal R_Z(\widehat e)}+\epsilon_{\rm rec}+2Bq,
\end{equation}
where $q\le\E[-\log p_\theta(\widehat e(T)\mid X)]/\log2$.
\end{proposition}

\emph{Proof idea.} Both equalities in Eq.~\eqref{eq:decomposition} are Pythagorean identities of conditional expectation, because $(M,C)$ is coarser than both $(M,Z)$ and $(T,C)$; Eq.~\eqref{eq:selection} follows by comparing $\widehat e$ with any codebook through $\mathcal R_Z=\mathcal E_{\rm abs}+\mathcal D-\mathcal G_C$. For Proposition~\ref{prop:decision}, the Bayes decoder's regret with respect to the teacher is at most the expected deviation of the teacher values from their conditional mean at the true best action, which after centering is at most the maximum centered reconstruction error and hence, by the Euclidean and Jensen inequalities, at most $s_Q\sqrt{|\mathcal A|\mathcal R_Z}$; replacing $U$ by $u$ twice costs $2\epsilon_Q$, and a wrong symbol costs at most $2B$. An input whose $\arg\max$ symbol differs from its label assigns the label a probability of at most one half and so contributes at least $\log2$ to the cross-entropy; Markov's inequality then bounds $q$.

Together, the propositions state what the construction controls: the conditional distortion, which offline RAVEN minimizes exactly, bounds the decision loss caused by compressing the sender's input into four symbols, up to estimation error and the residual gap of finitely many conditions. Section~\ref{sec:why} confirms this link empirically. For the online estimator, QMIX's monotonic mixer preserves the ordering of a receiver's utilities at fixed state and teammate utilities, and to first order $\bar V_i^{\dep}-\bar V_i^{\refb}=P\operatorname{diag}(\eta_{ia})\,\delta Q_i+O(\|\delta Q_i\|^2)$, where $\eta_{ia}$ are the mixer's slopes and $\delta Q_i=Q_i^{\dep}-Q_i^{\refb}$; $\Lcv$ thus drives the deployed utilities toward the reference ones along the directions that matter for team value, at a cost of $2\sum_i|\mathcal A_i|$ mixer evaluations per probe instead of an enumeration of joint actions.

\section{Experiments}\label{sec:experiments}
We ask four questions. Does RAVEN outperform recent communication methods under the same information boundary (Section~\ref{sec:nav})? Why does it work (Section~\ref{sec:why})? Does the teacher-free estimator carry the benefit to SMAC and MPE, and is the gain due to communication (Section~\ref{sec:scale})? What does it cost (Section~\ref{sec:cost})?

\subsection{Setup}\label{sec:setup}
\textbf{Tasks.} \emph{Navigation}: N1-C2 is MPE \texttt{simple\_reference}, in which each agent knows only its partner's goal; N1-C3/C4/C6 extend it to rings of three, four and six agents in which each agent's goal is known only to its predecessor; N2-L1/L2/L4/L6 are speaker--listener tasks with one, two, four and six listeners (25 steps each). \emph{Predator--prey} (PP1): three agents on a $5\times5$ grid with vision radius one must reach the prey together within 20 steps. \emph{SMAC}~\cite{r39}: 3m, 8m, MMM, 3s5z (hard) and MMM2 (super hard). \emph{MPE}: Spread, Tag and Crypto; in Crypto the eavesdropper Eve is controlled by the learned policies like the other agents, and her reward is included in the team return. All methods use the same message edges, the same one-step delay and the same execution-time information.

\textbf{Baselines.} On navigation we compare with five communication methods, each trained for 1.2M environment steps. SLIM~\cite{r8}, NDQ~\cite{r9} (three floats per message), CACOM~\cite{r5} (a 4-value request and a gated 12-value reply at 2~bits per value) and ExpoComm~\cite{r6} (64 floats) are run with their authors' code, architectures, optimizers and default hyperparameters; only the message edges, the one-step delay and the budget are adapted, and no task-specific tuning is performed. MACC~\cite{r20} is implemented from its paper. SLIM and MACC use the same four symbols as RAVEN. MACC diverges under the enforced one-step delay (returns between $-38$ and $-166$): its communication critic credits a message within the step in which it is sent, so a delayed message is credited to the wrong action. We therefore report MACC in Appendix~\ref{app:nav} (Table~\ref{stab:nav-main}) rather than in Table~\ref{tab:nav}. On SMAC and MPE, online RAVEN is compared at 2.05M steps with 13 methods, including QMIX~\cite{r14}, QPLEX~\cite{r40} and the QMIX-based communication methods TeamComm~\cite{r41}, TGCNet~\cite{r42}, CACOM and a single-neighbor ExpoComm. We run all 13 from the authors' official open-source implementations and take no numbers from prior publications. Online RAVEN uses one set of hyperparameters for all tasks. Two controls change only the communication pathway: a same-backbone QMIX without codec, receiver and reference, and online RAVEN without the reference-branch TD loss.

\textbf{Protocol.} Every RAVEN model, navigation baseline and control is evaluated at its final checkpoint on 1{,}000 held-out episodes. The training seed is the statistical unit (five seeds, four on SMAC; ten in the ablation): we report means, standard deviations and 95\% $t$-intervals, and paired two-sided $t$-tests with Holm correction within each family of comparisons.

\subsection{Comparison with Communication Methods}\label{sec:nav}
\begin{table*}[t]
\centering
\caption{Navigation: team return at the final checkpoint (mean$_{\pm\text{s.d.}}$ over 5 seeds, 1{,}000 held-out episodes per model; higher is better). Best per row in bold, second underlined (reference columns excluded). $^{\ddagger}$Trained with one listener and deployed unchanged. The last row gives the bits per message. MACC, which diverges under the enforced one-step delay (Section~\ref{sec:setup}), is reported in Table~\ref{stab:nav-main}. The reference columns hold the teacher fixed on N1-C2 (ten-seed controlled study of Section~\ref{sec:why}): the centralized teacher itself ($\widehat Q$, evaluated on its selection episodes, hence an optimistic upper bound), receivers supervised by the same teacher without messages (no msg.), and the same supervision with a learned four-symbol channel in place of the constructed codebook (learn.\ ch.).}
\label{tab:nav}
\setlength{\tabcolsep}{3.2pt}
\small
\begin{tabular}{@{}lcccccc|ccc@{}}
\toprule
Setting & \textbf{\off{}} & \textbf{\on{}} & SLIM & NDQ & CACOM & ExpoComm & $\widehat Q$ (ub) & no msg. & learn.\ ch. \\
\midrule
N1-C2 (2 agents, pair) & \textbf{$-$9.59}$_{\pm0.60}$ & \underline{$-$12.35}$_{\pm1.68}$ & $-$18.81$_{\pm0.76}$ & $-$21.34$_{\pm1.87}$ & $-$17.25$_{\pm0.97}$ & $-$14.98$_{\pm1.44}$ & $-$8.52$_{\pm0.16}$ & $-$19.27$_{\pm0.19}$ & $-$14.49$_{\pm0.65}$ \\
N1-C3 (3 agents, ring) & \textbf{$-$11.47}$_{\pm0.39}$ & \underline{$-$15.21}$_{\pm2.78}$ & $-$18.79$_{\pm0.71}$ & $-$19.98$_{\pm0.95}$ & $-$17.84$_{\pm0.15}$ & $-$17.35$_{\pm0.92}$ & -- & -- & -- \\
N1-C4 (4 agents, ring) & \textbf{$-$13.74}$_{\pm0.64}$ & \underline{$-$17.86}$_{\pm0.30}$ & $-$20.02$_{\pm0.85}$ & $-$20.06$_{\pm0.67}$ & $-$17.89$_{\pm0.44}$ & $-$17.92$_{\pm0.30}$ & -- & -- & -- \\
N1-C6 (6 agents, ring) & $-$18.15$_{\pm0.12}$ & $-$17.96$_{\pm0.23}$ & $-$20.13$_{\pm0.72}$ & $-$24.32$_{\pm3.55}$ & \textbf{$-$17.78}$_{\pm0.28}$ & \underline{$-$17.94}$_{\pm0.26}$ & -- & -- & -- \\
\midrule
N2-L1 (1 listener) & \textbf{$-$7.70}$_{\pm0.17}$ & \underline{$-$7.86}$_{\pm0.24}$ & $-$11.89$_{\pm3.01}$ & $-$14.19$_{\pm3.02}$ & $-$12.64$_{\pm4.32}$ & $-$15.42$_{\pm0.28}$ & -- & -- & -- \\
N2-L2 (2 listeners) & \textbf{$-$7.86}$_{\pm0.14}{}^{\ddagger}$ & \underline{$-$9.13}$_{\pm1.87}$ & $-$12.10$_{\pm2.92}{}^{\ddagger}$ & $-$15.08$_{\pm2.78}$ & $-$9.65$_{\pm3.58}$ & $-$14.16$_{\pm2.39}$ & -- & -- & -- \\
N2-L4 (4 listeners) & \textbf{$-$7.92}$_{\pm0.21}{}^{\ddagger}$ & \underline{$-$9.82}$_{\pm2.31}$ & $-$12.19$_{\pm2.83}{}^{\ddagger}$ & $-$20.52$_{\pm1.21}$ & $-$10.68$_{\pm3.61}$ & $-$13.66$_{\pm3.09}$ & -- & -- & -- \\
N2-L6 (6 listeners) & \textbf{$-$7.82}$_{\pm0.09}{}^{\ddagger}$ & \underline{$-$9.23}$_{\pm0.53}$ & $-$12.26$_{\pm2.99}{}^{\ddagger}$ & $-$24.60$_{\pm5.51}$ & $-$13.81$_{\pm5.28}$ & $-$13.98$_{\pm2.40}$ & -- & -- & -- \\
\midrule
Bits per message & 2 & 2 & 2 & 96 & 24 (+8 request) & 2{,}048 & -- & -- & 2 \\
\bottomrule
\end{tabular}
\end{table*}

Table~\ref{tab:nav} gives the final returns. Offline RAVEN (\off{}) has the highest return in seven of the eight settings, and online RAVEN (\on{}) is second in the same seven. Seed-paired tests confirm the ranking: of the 40 comparisons between \off{} and the five external methods, 38 favor \off{}, 28 remain significant after Holm correction, and in 31 every one of the five seeds favors \off{}. On N1-C2, \off{} beats every external method on every seed. Its margin over the strongest external method is $+5.39$, $+5.88$ and $+4.15$ on the pair and on the three- and four-agent rings, and a confirmation on fresh seeds that took no part in method development reproduces its leads over SLIM on the three rings ($+7.32$, $+6.28$, $+1.97$; 5/5 seeds each, Holm $p\le0.0021$).

The reference columns of Table~\ref{tab:nav} show where this advantage comes from. Offline RAVEN's receivers are trained by maximizing the teacher's joint value (Eq.~\eqref{eq:receiver}), so one might suspect that the teacher, rather than the communication principle, produces the gap. Holding the teacher fixed on N1-C2, receivers that are supervised identically but receive no messages reach only $-19.27$, behind SLIM, CACOM and ExpoComm, and replacing the constructed codebook by a learned channel of the same alphabet under the same supervision gives $-14.49$: teacher supervision alone does not beat reward-trained communication, and it is the exact conditional construction that turns the supervision into a lead. The comparison is also budget-matched. The teacher's 1.2M training steps are the only environment interaction of offline RAVEN---the construction data are the teacher's own transitions, and sender distillation and receiver training reuse these rows without new environment steps---so the 1.2M budget of Table~\ref{tab:nav} already includes the teacher; checkpoint selection and final evaluation use disjoint sets of episodes.

\begin{figure}[t]
\centering
\includegraphics[width=\columnwidth]{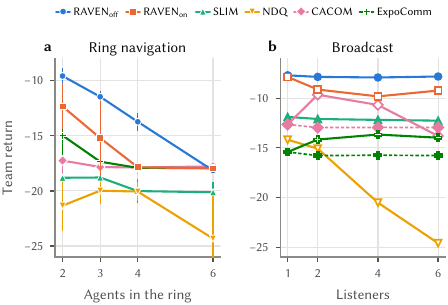}
\caption{Final return against team size (mean and 95\% interval over 5 seeds). (a) Rings of two to six agents. (b) Broadcast to one to six listeners: filled markers are models trained with one listener and deployed unchanged (dashed: CACOM and ExpoComm transferred in the same way), hollow markers are retrained at every size.}
\label{fig:scaling}
\Description{Left: in ring navigation every method has a lower return with six agents than with two; offline RAVEN is highest for two to four agents and all methods except NDQ and SLIM meet near minus 18 at six agents. Right: in broadcast, offline RAVEN deployed zero-shot stays near minus 7.8 for all listener counts, above online RAVEN, CACOM, SLIM, ExpoComm and NDQ, the last of which degrades to minus 24.6 at six listeners.}
\end{figure}

The broadcast family shows a second strength of RAVEN, its reliability: \off{} solves the task on every seed, with standard deviations of 0.09--0.21 against 2.8--5.3 for SLIM and CACOM, whose seeds split between runs that solve the task and runs that do not. Because its codebook is shared by all recipients, a model trained with a single listener can be deployed unchanged to two, four or six listeners and loses at most 0.22 return (Figure~\ref{fig:scaling}b); CACOM and ExpoComm transferred in the same way stay near $-12.95$ and $-15.77$, and even when retrained at every size neither reaches the zero-shot \off{}, while NDQ degrades from $-14.19$ to $-24.60$. In the rings (Figure~\ref{fig:scaling}a), \off{} leads every external method from two to four agents. The teacher-free estimator keeps most of this benefit. \on{} beats NDQ and MACC in all 16 of its paired comparisons with them and is ahead of CACOM and of ExpoComm in six of the eight settings and within one standard deviation of them in the other two (N1-C4 and N1-C6), while sending 12--1{,}024$\times$ fewer bits; it trails the exact construction by 0.2--4.1 return in seven settings and is slightly ahead of it on the six-agent ring.

\subsection{Why RAVEN Works}\label{sec:why}
\begin{figure*}[t]
\centering
\includegraphics[width=\textwidth]{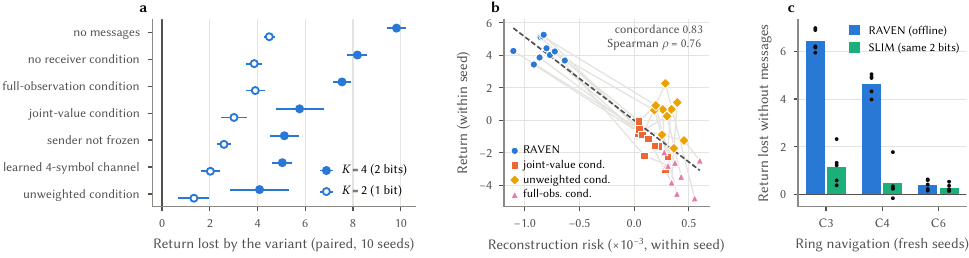}
\caption{Mechanism of offline RAVEN on navigation. (a) Return lost by each variant relative to full RAVEN on N1-C2 (paired mean and 95\% interval over 10 independently trained seeds; filled: 4 symbols, hollow: 2 symbols). (b) Held-out reconstruction risk $\mathcal R$ against return for RAVEN and three condition variants, both centered within seed (grey lines join the methods of one seed). (c) Return lost when incoming messages are removed at test time (fresh seeds; dots are seeds); SLIM uses the same 2-bit budget.}
\label{fig:mechanism}
\Description{Three panels. Left: a forest plot in which every variant loses return relative to full RAVEN, the largest losses coming from removing messages and from removing the receiver condition. Middle: a scatter plot with a negative trend between reconstruction risk and return, RAVEN points at the lowest risk and highest return. Right: bars showing that RAVEN loses much more return than SLIM when its messages are removed on the three- and four-agent rings.}
\end{figure*}

\begin{table*}[t]
\centering
\caption{SMAC (win rate, \%) and MPE (team return) at 2.05M environment steps. \on{}: mean$_{\pm\text{s.d.}}$ with 1{,}000 evaluation episodes per model (SMAC: 4 seeds; MPE: 5 seeds). Reference methods: our runs of the authors' official open-source code under the same budget; mean over runs. ``Norm.'' and ``rank'' average the per-task min--max normalized score and the rank over all 14 methods on the eight tasks (the remaining seven methods are listed in Table~\ref{stab:smac-mpe}). $^{\dagger}$Single-neighbor variant. Best per column in bold, second underlined.}
\label{tab:smac}
\setlength{\tabcolsep}{3.3pt}
\small
\begin{tabular}{@{}lccccc|ccc|ccc@{}}
\toprule
Method & 3m & 8m & MMM & MMM2 & 3s5z & Spread & Tag & Crypto & Bits/msg & Norm. & Rank \\
\midrule
\textbf{\on{}} & 99.72$_{\pm0.22}$ & \textbf{99.95}$_{\pm0.06}$ & \textbf{100.0}$_{\pm0.0}$ & \underline{87.35}$_{\pm8.30}$ & \textbf{97.80}$_{\pm1.12}$ & \underline{$-$29.00}$_{\pm0.59}$ & \textbf{277.3}$_{\pm3.0}$ & \underline{48.03}$_{\pm0.04}$ & \textbf{2} & \textbf{0.98} & \textbf{1.75} \\
\midrule
QMIX~\cite{r14} & \underline{99.80} & 98.80 & 98.60 & 57.80 & 86.40 & $-$43.42 & 23.39 & 0.38 & 0 & 0.74 & 7.3 \\
QPLEX~\cite{r40} & 99.75 & \underline{99.50} & 30.75 & 62.00 & 96.50 & $-$31.57 & 235.1 & 45.44 & 0 & 0.80 & 5.0 \\
TeamComm~\cite{r41} & 99.25 & 99.00 & 98.00 & 64.25 & 91.00 & $-$43.73 & 68.56 & 48.00 & $\ge$2{,}048 & 0.84 & 5.6 \\
TGCNet~\cite{r42} & 99.00 & 99.25 & 74.00 & 66.50 & 96.75 & $-$44.95 & 55.46 & 48.00 & $\ge$2{,}048 & 0.80 & 5.9 \\
CACOM~\cite{r5} & 96.25 & 43.50 & 98.75 & 34.25 & \underline{97.25} & $-$48.47 & 251.1 & 3.18 & $\ge$16 & 0.67 & 8.1 \\
ExpoComm$^{\dagger}$~\cite{r6} & 99.00 & \underline{99.50} & \textbf{100.0} & \textbf{94.50} & 96.25 & \textbf{$-$26.80} & \underline{276.9} & $-$31.20 & 4{,}096 & \underline{0.87} & \underline{4.0} \\
\bottomrule
\end{tabular}
\end{table*}

\textbf{Conditioning on the receiver is the key design choice.} The ablation in Figure~\ref{fig:mechanism}a trains ten independent seeds per variant on N1-C2; each seed has its own teacher, and all variants of a seed share it, the network capacity and the total number of updates, so each variant changes one factor. Communication is worth $+9.82$ return. Compressing the \emph{marginal} value instead of the conditional one loses $8.18$ of it (83\%) on all ten seeds; this variant keeps the exact construction but scores symbols by the marginal profile, the quantity that value-clustering constructions optimize~\cite{r3,r4}, and its collapse is the cancellation illustrated in Figure~\ref{fig:example}. The way conditions are formed also matters: conditions built from the receiver's full observation ($-7.54$), from the sensitivity of the joint value instead of the receiver's own action gaps ($-5.76$), or from unweighted clusters ($-4.09$) are all significantly worse. A condition that distinguishes everything the receiver sees spends the four symbols on distinctions that do not change its decision. Construction and freezing also matter, and they address the second failure mode: a learnable four-symbol channel of the same capacity loses $5.04$, and continuing to update the distilled sender together with the receivers loses $5.12$. Both variants keep the receiver-conditioned supervision but train through the discrete channel, as end-to-end methods do, and both forfeit about half of the communication gain, consistent with biased gradients through the discrete channel eroding the code. All seven differences are positive on 10/10 seeds with Holm-adjusted $p<3\times10^{-4}$; with two symbols the ordering is largely preserved at about half the scale.

\textbf{The construction's objective predicts return.} For each seed we computed, on held-out samples, the conditional reconstruction risk that the codebook minimizes (Proposition~\ref{prop:selection}). A lower risk agrees with a higher return in 83\% of within-seed method pairs (Spearman $\rho=0.76$; Figure~\ref{fig:mechanism}b), and RAVEN has the lowest risk on every seed. Restricted to the three condition variants, the compression term alone orders their returns with 87\% concordance ($\rho=0.80$): among these variants, return tracks how well the four symbols preserve the conditional values.

\textbf{The receivers rely on the symbols.} Removing incoming messages at test time costs \off{} 6.44 and 4.64 return on the three- and four-agent rings, against 1.17 and 0.49 for SLIM at the same budget (Figure~\ref{fig:mechanism}c), and the empirical symbol entropy is 1.94 of 2 bits, so all four symbols are used. Because the symbols are discrete and their meaning is fixed, the frozen models also degrade gracefully under channel constraints: sending a fresh symbol only every fourth step changes the return by $-0.009$ (95\% interval $[-0.022,0.004]$) while cutting the serialized payload by 72\%, and 10\% or 30\% random message loss moves the return only from $-9.78$ to $-9.80$ and $-9.88$.

\subsection{Online RAVEN on SMAC and MPE}\label{sec:scale}
\begin{figure}[t]
\centering
\includegraphics[width=\columnwidth]{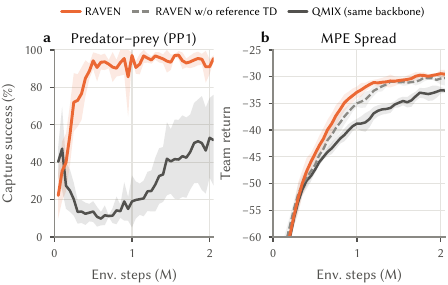}
\caption{Online RAVEN against controls that change only the communication pathway (mean $\pm$ s.d.; 100 monitoring episodes per point). (a) Predator--prey capture success, 5 seeds. (b) MPE Spread team return, 5 seeds (variant without the reference TD loss: 3 matched seeds).}
\label{fig:online}
\Description{Left: capture success of online RAVEN rises to about 85 percent within 0.4 million steps and above 90 percent by 0.5 million steps, while the same-backbone QMIX falls to about 12 percent and recovers only to about 50 percent. Right: on Spread, online RAVEN is above the variant without the reference TD loss, which is above the same-backbone QMIX, during most of training.}
\end{figure}

\begin{table}[t]
\centering
\caption{Resources saved by RAVEN, $100(1-\text{RAVEN}/\text{other})$ in \%, as the range over settings (navigation) or tasks (SMAC and MPE). Parameters and execution FLOPs per team step are counted on the networks used in the runs; training time is compared on the same machine. $^{\star}$Lower bound (only part of the message counted).}
\label{tab:cost}
\setlength{\tabcolsep}{4pt}
\small
\begin{tabular}{@{}lcccc@{}}
\toprule
vs. & Bits/msg & FLOPs & Params & Train time \\
\midrule
\multicolumn{5}{@{}l}{\textit{\off{}, navigation}} \\
SLIM & same & -- & 95.6--98.0 & 39.6--55.4 \\
NDQ & 97.9 & 82.2--96.8 & $-$1.4\,--\,85.1 & -- \\
CACOM & 91.7 & 85.8--97.7 & $-$2.1\,--\,85.0 & -- \\
ExpoComm & 99.9 & 91.5--98.3 & 48.4--92.3 & -- \\
\midrule
\multicolumn{5}{@{}l}{\textit{\on{}, navigation}} \\
CACOM & 91.7 & $-$18.6\,--\,1.6 & $-$48.3\,--\,$-$38.2 & 70.6--74.3 \\
\midrule
\multicolumn{5}{@{}l}{\textit{\on{}, SMAC and MPE}} \\
TeamComm & 99.9$^{\star}$ & 4.7--21.0 & 23.7--29.6 & -- \\
TGCNet & 99.9$^{\star}$ & $-$2.9\,--\,13.7 & 33.6--40.2 & -- \\
ExpoComm$^{\dagger}$ & 99.95 & 78.5--82.1 & 78.8--82.1 & -- \\
\bottomrule
\end{tabular}
\end{table}

\textbf{The gain comes from communication.} Holding the recurrent backbone, optimizer, replay, exploration, budget and initialization fixed and removing only the channel isolates its contribution. In predator--prey, where three agents with vision radius one must agree on a common target, the 2-bit channel raises capture success from $53.20\%$ to $96.04\%$, $+42.84$ percentage points (95\% interval $[15.5,70.2]$, 5/5 seeds, Holm-adjusted $p=0.024$). RAVEN reaches about 85\% success after 0.4M steps and ends between 94.5\% and 97.0\% on every seed, whereas the same-backbone QMIX collapses to 10--14\% before recovering slowly (Figure~\ref{fig:online}a). On Spread, where every agent already sees all landmarks, the channel still adds $+2.87$ return ($[1.77,3.97]$, 5/5 seeds, $p=0.002$). Removing only the reference-branch TD loss costs $0.69$ (3/3 seeds, $p=0.043$) and lowers the return by about three units between 0.7M and 1.1M steps (Figure~\ref{fig:online}b): training the reference branch speeds up learning of the deployed channel.

\textbf{The teacher-free estimator is competitive on SMAC and MPE.} On SMAC (Table~\ref{tab:smac}), \on{} wins 99.7--100\% of the games on 3m, 8m and MMM and has the highest win rate on the hard 3s5z (97.80\%). On the super-hard MMM2, value factorization alone reaches 57.8\% (QMIX) and 62.0\% (QPLEX) and the QMIX-based communication methods 64.3\% (TeamComm) and 66.5\% (TGCNet); \on{} reaches 87.35\%, at least 20.8 points higher, with 2-bit messages; only the single-neighbor ExpoComm is higher. It also learns MMM fastest of all methods, with a time-averaged win rate of 92.8\% against 90.8\% for the next best. On MPE it is best on Tag and second on Spread and Crypto, behind the single-neighbor ExpoComm and an attention-based QMIX variant, respectively (Table~\ref{stab:smac-mpe} lists all 14 methods). Over the eight tasks it has the highest mean normalized score (0.98) and the best mean rank (1.75) of the 14 methods, is in the top three on seven tasks, and never falls below a normalized score of 0.92, whereas every other method falls below 0.4 on at least one task. The single-neighbor ExpoComm, which sends 4{,}096 bits per message, is the closest competitor (0.87).

\subsection{Cost}\label{sec:cost}

Every RAVEN message is 2 bits (Table~\ref{tab:cost}). Over complete training runs on navigation, \on{} delivers 97.9\% fewer bits than NDQ, 79.3--90.0\% fewer than CACOM, whose gate suppresses some replies, and 99.8\% fewer than ExpoComm. The offline implementation is also lightweight: its execution FLOPs are 82--98\% below those of NDQ, CACOM and ExpoComm, its broadcast model has 5{,}257 parameters, 98\% fewer than SLIM's, and, including its teacher, it trains in 40--55\% less time than SLIM on the same server in the ring settings. The online implementation needs 71--74\% less training time than CACOM on the same machine and, on SMAC and MPE, needs 24--40\% fewer parameters than TeamComm and TGCNet and about 80\% fewer parameters and FLOPs than the single-neighbor ExpoComm.

\section{Discussion and Limitations}\label{sec:discussion}
RAVEN's advantage is largest when a few symbols must carry decision-relevant information that the receiver cannot observe---pairs, small rings and broadcast---and where the receiver's context changes the meaning of that information. Several limitations remain. First, in the six-agent ring CACOM and ExpoComm are ahead of \off{} by 0.37 and 0.21 return; removing all messages costs \off{} only 0.41 there, so a one-hop ring carries little decision value at this size, and multi-hop relaying is a natural extension. Second, the online implementation learns more slowly early in training on the hardest maps: on 3s5z it reaches 80\% win rate after 1.2M steps, later than QPLEX, and on MMM2 the single-neighbor ExpoComm, with 2{,}048 times more bits per message, is higher (94.5\%). Third, the online implementation has 24--54\% more parameters than CACOM on SMAC and MPE (Table~\ref{stab:efficiency}), so its savings there are in bandwidth rather than network size. Finally, the offline implementation requires a teacher and an enumerable joint action space for the receiver objective; the online implementation removes both requirements.

\section{Conclusion}
RAVEN rests on one principle: a finite-alphabet message should preserve the receiver's action-value differences in the receiver's own situation. The offline estimator realizes it exactly and gives the strongest results where a teacher is available; the online estimator carries it, without a teacher, to ten-agent StarCraft teams. In both, two bits per message, interpreted by each receiver with its private information, are enough to outperform or match methods that exchange tens to thousands of bits. Code is available at \url{https://github.com/sswun/RAVEN}.

\bibliographystyle{ACM-Reference-Format}
\bibliography{references}

\clearpage
\appendix
\section*{Appendix overview}
The appendices list the notation (Appendix~\ref{app:notation}), prove the statements of Section~\ref{sec:analysis} (Appendix~\ref{app:proofs}), give the implementation of both estimators (Appendices~\ref{app:offline} and~\ref{app:online}), describe tasks, baselines and the statistical protocol (Appendix~\ref{app:setup}), and report the complete results behind every number in the paper (Appendices~\ref{app:nav}--\ref{app:cost}). All results are computed from the final checkpoints of the runs described above. Code is available at \url{https://github.com/sswun/RAVEN}; the result files and the scripts that produce every table and figure from them will be released there as well.

\section{Notation}\label{app:notation}
Table~\ref{stab:notation} lists the symbols used throughout the paper.

\begin{table}[h]
\centering
\caption{Notation.}
\label{stab:notation}
\small
\setlength{\tabcolsep}{3pt}
\begin{tabularx}{\columnwidth}{@{}l>{\raggedright\arraybackslash}X@{}}
\toprule
Symbol & Meaning \\
\midrule
$N$, $K=4$ & number of agents; alphabet size (2 bits per message) \\
$X_j^t$, $M_j^t$ & permitted local input and symbol of sender $j$ at step $t$ \\
$Z_i^{t}$, $\mathbf M_i^{t}$ & local information and mailbox of receiver $i$ \\
$\mathcal A_i$, $P$ & legal actions of receiver $i$; centering matrix $I-\mathbf 1\mathbf 1^\top/|\mathcal A_i|$ \\
$\|v\|_{\mathcal A}^2$ & $|\mathcal A|^{-1}\sum_a v_a^2$ \\
\midrule
$\widehat Q$, $\xi_n$, $\rho$ & teacher, its input at the arrival time, partner-action reference \\
$u_{na}$, $Y_{na}$, $s_Q$ & receiver action value, centered and scaled target, global scale \\
$\mathcal D$ & construction data (rows $n$) \\
$T=\tau(X)$, $n_T$ & sender type and number of types (12) \\
$V$, $\bar V\in\mathbb R^F$, $C$, $n_C$ & receiver features, standardized features, condition, number of conditions (12) \\
$\zeta_\ell$, $\Omega$ & action-gap sensitivity of feature $\ell$; condition metric \\
$N_{rca}$, $S_{rca}$, $\widetilde\mu_{rca}$, $\alpha_{rca}$ & cell mass, target sum, shrunk mean, empirical mass \\
$e$, $\widehat e$, $\widehat D(e)$ & codebook (types $\to$ symbols), selected codebook, empirical distortion \\
$s_\theta$, $p_\theta$, $\pi_\phi$ & frozen sender, its symbol distribution, offline receiver policy \\
\midrule
$h_i^t$, $q_i^{{\rm loc},t}$ & recurrent state and local utilities (online) \\
$l_j^t$, $\widetilde m_j^t$, $\Phi$, $g_j^t$ & encoder logits, straight-through one-hot, $4\times64$ symbol embedding, decoded symbol \\
$\mathcal N_i^\star$, $\beta_{ij}$ & two selected senders and their attention weights \\
$c_i^b$, $\Delta_i^b$, $Q_i^b$ & context, correction and utilities of branch $b\in\{\dep,\refb\}$ \\
$f_\psi$, $V_i^b$, $\bar V_i^b$ & QMIX mixer, mixed receiver profile, centered profile \\
$\mathcal B$, $\sigma_i^2$ & probe set, scale of $\Lcv$ \\
\midrule
$\mathcal R_C,\mathcal R_Z,\mathcal G_C$ & reconstruction risks and the gap between them \\
$\mathcal E_{\rm abs}$, $\mathcal D(e)$ & abstraction and compression errors \\
\bottomrule
\end{tabularx}
\end{table}

\section{Proofs}\label{app:proofs}
Throughout, the legal set is fixed with $d=|\mathcal A|$ actions, $\|v\|_{\mathcal A}^2=d^{-1}\sum_av_a^2$, all random quantities are square integrable, and $M=e(T)$ with $T=\tau(X)$.

\subsection{Centering identity \texorpdfstring{(Eq.~\eqref{eq:center})}{}}
Let $w=u-v$ and $\bar w=d^{-1}\sum_aw_a$. Then $\|Pw\|_2^2=\sum_a(w_a-\bar w)^2=\frac{1}{2d}\sum_{a,b}(w_a-w_b)^2=\frac1d\sum_{a<b}(w_a-w_b)^2$, and $w_a-w_b=(u_a-u_b)-(v_a-v_b)$.

\subsection{Merge cost and exactness of the enumeration}
Fix a cell $(c,a)$ and two groups of types with masses $v_1,v_2$ and weighted means $\mu_1,\mu_2$. The weighted sum of squared deviations of the union about its mean equals the sum of the two within-group sums plus $\frac{v_1v_2}{v_1+v_2}(\mu_1-\mu_2)^2$ (the Ward identity). Summing over cells gives the merge cost $\Delta$ of Section~\ref{sec:offline}. Consequently (i) $\Delta\ge0$, so splitting a group never increases $\widehat D$, and (ii) $\Delta=0$ if and only if the two groups have equal means in every cell where both have mass; in particular, groups that never have mass in the same condition merge for free.

Let $\mathcal E_{\le K}$ be the codebooks with at most $K$ nonempty groups over $n_T\ge K$ types and $\mathcal E_{=K}$ those with exactly $K$. By (i), every $e\in\mathcal E_{\le K}$ can be refined to some $e'\in\mathcal E_{=K}$ with $\widehat D(e')\le\widehat D(e)$, so $\min_{\mathcal E_{=K}}\widehat D=\min_{\mathcal E_{\le K}}\widehat D$. Symbols are exchangeable, so it suffices to enumerate unlabeled partitions; their number is the Stirling number
\begin{equation}
S(12,4)=\frac{4^{12}-4\cdot3^{12}+6\cdot2^{12}-4}{24}=611{,}501 .
\end{equation}
Ties are broken by a fixed order of the partitions. For several recipients, $\widehat D$ is the equally weighted sum of the recipients' tables under one shared map, and the same argument applies.

\subsection{Proof of Proposition~\texorpdfstring{\ref{prop:selection}}{3.1}}
\emph{First equality.} Because $C=c(V)$ and $V\subseteq Z$, the $\sigma$-algebra of $(M,C)$ is contained in that of $(M,Z)$. Write $Y-\E[Y\mid M,C]=A+B$ with $A=Y-\E[Y\mid M,Z]$ and $B=\E[Y\mid M,Z]-\E[Y\mid M,C]$. $B$ is $(M,Z)$-measurable and $\E[A\mid M,Z]=0$, so $\E\langle A,B\rangle_{\mathcal A}=0$ and $\mathcal R_C=\mathcal R_Z+\mathcal G_C$.

\emph{Second equality.} Because $M=e(T)$, the $\sigma$-algebra of $(M,C)$ is contained in that of $(T,C)$. With $\mu=\E[Y\mid T,C]$, write $Y-\E[Y\mid M,C]=(Y-\mu)+(\mu-\E[Y\mid M,C])$; the second term is $(T,C)$-measurable and $\E[Y-\mu\mid T,C]=0$, so $\mathcal R_C=\mathcal E_{\rm abs}+\mathcal D(e)$. $\mathcal E_{\rm abs}$ does not depend on $e$.

\emph{Selection bound.} Combining both equalities, $\mathcal R_Z(e)=\mathcal E_{\rm abs}+\mathcal D(e)-\mathcal G_C(e)$ for every $e$. For any codebook $e'$,
\begin{align*}
\mathcal R_Z(\widehat e)-\mathcal R_Z(e')
&=\mathcal D(\widehat e)-\mathcal D(e')+\mathcal G_C(e')-\mathcal G_C(\widehat e)\\
&\le\big[\widehat D(\widehat e)+\epsilon_D\big]-\big[\widehat D(e')-\epsilon_D\big]+\operatorname{osc}_e\mathcal G_C(e)\\
&\le\epsilon_{\rm opt}+2\epsilon_D+\operatorname{osc}_e\mathcal G_C(e),
\end{align*}
using $\widehat D(\widehat e)\le\inf_e\widehat D(e)+\epsilon_{\rm opt}\le\widehat D(e')+\epsilon_{\rm opt}$. Taking the infimum over $e'$ gives Eq.~\eqref{eq:selection}. The empirical distortion $\widehat D$ weights each row's legal actions by $1/|\mathcal A_n|$, which is the norm $\|\cdot\|_{\mathcal A}$; $\epsilon_D$ therefore covers finite-sample and shrinkage effects only. \hfill$\square$

\subsection{Proof of Proposition~\texorpdfstring{\ref{prop:decision}}{3.2}}
Let $\hat u=\E[u\mid M,Z]$ (a vector over actions), let the Bayes decoder choose $a^\star=\arg\max_a\hat u_a$, and let $a^\circ=\arg\max_au_a$. Then
\begin{equation}
\begin{aligned}
\E[u_{a^\circ}-u_{a^\star}]={}&\E[u_{a^\circ}-\hat u_{a^\circ}]+\E[\hat u_{a^\circ}-\hat u_{a^\star}]\\
&+\E[\hat u_{a^\star}-u_{a^\star}] .
\end{aligned}
\end{equation}
The middle term is non-positive by the definition of $a^\star$, and the last is zero because $a^\star$ is $(M,Z)$-measurable. With $\delta=u-\hat u$ and $\bar\delta=d^{-1}\sum_a\delta_a$, $\delta_{a^\circ}=(P\delta)_{a^\circ}+\bar\delta$ and $\E\bar\delta=\E\,\E[\bar\delta\mid M,Z]=0$. Hence
\begin{equation}
\begin{aligned}
\E[u_{a^\circ}-u_{a^\star}]&\le\E\max_a|(P\delta)_a|\le\E\|P\delta\|_2\\
&=\sqrt d\,\E\|P\delta\|_{\mathcal A}\le\sqrt{d\,\E\|P\delta\|_{\mathcal A}^2}.
\end{aligned}
\end{equation}
For a fixed legal set, centering commutes with conditional expectation, $Pu=s_QY$ and $P\hat u=s_Q\E[Y\mid M,Z]$, so $\E\|P\delta\|_{\mathcal A}^2=s_Q^2\,\mathcal R_Z(\widehat e)$ and the teacher regret of the Bayes decoder is at most $s_Q\sqrt{d\,\mathcal R_Z(\widehat e)}$.

For the true payoffs, $\E\max_aU_a-\E\max_au_a\le\epsilon_Q$ and $\E u_{a^\star}-\E U_{a^\star}\le\epsilon_Q$, so the Bayes decoder loses at most $2\epsilon_Q+s_Q\sqrt{d\,\mathcal R_Z(\widehat e)}$ in true payoff. The receiver loses at most $\epsilon_{\rm rec}$ more under the ideal symbol. The deployed symbol $s_\theta(X)$ differs from $\widehat e(T)$ with probability $q$, and on that event the payoff changes by at most $2B$; this gives Eq.~\eqref{eq:decision}. Finally, if $\arg\max_mp_\theta(m\mid X)\neq\widehat e(T)$ then $p_\theta(\widehat e(T)\mid X)\le\frac12$, i.e.\ $-\log p_\theta(\widehat e(T)\mid X)\ge\log2$, and Markov's inequality gives $q\le\E[-\log p_\theta(\widehat e(T)\mid X)]/\log2$. \hfill$\square$

\subsection{Online estimator}
For fixed state $s$ and teammate utilities, the QMIX mixer is non-decreasing in each agent's utility, so $Q_i^b(a)\ge Q_i^b(a')$ implies $V_i^b(a)\ge V_i^b(a')$: mixing preserves the ranking of a receiver's actions. If the mixer has a locally Lipschitz gradient, a first-order expansion around the reference utility of each action gives $V_i^{\dep}(a)-V_i^{\refb}(a)=\eta_{ia}\,\delta Q_i(a)+O(\delta Q_i(a)^2)$ with slopes $\eta_{ia}\ge0$, and centering yields the relation stated in Section~\ref{sec:analysis}. Each probe evaluates the mixer once per legal action and branch, i.e.\ $2\sum_i|\mathcal A_i|$ times, instead of enumerating joint actions.

\section{Offline RAVEN: Implementation}\label{app:offline}
\textbf{Teachers.} Pair navigation (N1-C2, MPE \texttt{simple\_reference}; 11 local observation features, five movement actions, 25-step episodes) uses a Double-DQN teacher with ReLU layers $22\to128\to128\to25$ over joint actions. Ring navigation uses $11N\to128\to128\to5N$ with centralized-input additive heads $\widehat Q(x,\mathbf a)=N^{-1}\sum_ku_k(x,a_k)$~\cite{r13}; the speaker--listener teacher is $11\to128\to128\to5$. Teachers are trained for 1.2M environment steps with Adam~\cite{r38} at $3\times10^{-4}$, discount 0.95, a replay capacity of 200{,}000 and batches of 512; after 5{,}000 collection steps they update every 16 steps and synchronize targets every 1{,}000 updates, and exploration decays from 1 to 0.05 over 600{,}000 steps. Pair-navigation and speaker--listener teachers are selected on independent 200-episode sets among the checkpoints at 0, 0.1M, 0.3M, 0.6M and 1.2M steps; ring teachers use the final checkpoint.

\textbf{Construction.} The construction data consist of the first 10{,}000 transitions and 10{,}000 transitions sampled from the final replay, without terminal transitions; normalization, types, conditions and codebooks are all computed on these rows. The partner reference is uniform over the partner's actions in navigation; in speaker--listener tasks the speaker's action is held fixed. The scale $s_Q$ is the global standard deviation of teacher values (floored at $10^{-6}$); rows with a single legal action have a zero target. Types come from k-means on standardized sender inputs (speaker--listener tasks use the three target colors). Conditions use $n_C=12$, k-means++ seeding~\cite{r31}, 40 Lloyd iterations~\cite{r32} and the metric of Eq.~\eqref{eq:metric} with derivatives taken with the partner reference and all other teacher inputs held fixed; the construction is rejected if the sensitivity trace $\sum_\ell\zeta_\ell$ is below $10^{-16}$. The pseudo-count is $\lambda_{ra}=32N_{ra}/n_r$; unsupported type--action pairs are omitted and zero-mass cells contribute nothing. The codebook is the exact minimizer of Eq.~\eqref{eq:codebook} (Appendix~\ref{app:proofs}); in broadcast tasks one map serves all listeners, and each ring sender has one recipient.

\textbf{Sender and receivers.} Pair-navigation senders and receivers are Tanh networks $11\to64\to64\to4$ and $15\to32\to5$ (the receiver input is its observation and the one-hot symbol). The sender is fitted by 2{,}048 cross-entropy updates and then frozen (parameters, normalization statistics and all features that produce symbols); receivers are trained for 8{,}192 updates on Eq.~\eqref{eq:receiver} with batches of 512, Adam at 0.003 and gradient-norm clipping at 5. At execution the sender emits the $\arg\max$ symbol, which the receiver uses at the next step, and the receiver acts greedily. In the ablation, the learned four-symbol channel and the no-message policy receive 10{,}240 policy updates, the same total number of updates as RAVEN; all variants of a seed share its teacher and construction data.

\section{Online RAVEN: Implementation}\label{app:online}
\textbf{Networks.} Each agent's observation, previous action (one-hot) and identity pass through a 64-unit layer and a 64-unit GRU~\cite{r33}, followed by a linear local utility head; parameters are shared across agents. The encoder is linear ($64\to4$); transmission uses the hard $\arg\max$ and the straight-through estimator of Eq.~\eqref{eq:st} with temperature one and no Gumbel noise; $\Phi$ is a shared $4\times64$ embedding. Queries $W_Qh_i$ and keys $W_K\LN(g_j)$ have 32 dimensions and are compared by the dot product divided by $\sqrt{32}$; the self edge is excluded, the two highest-scoring senders are kept ($\min\{2,N-1\}$), and the softmax runs over the kept senders only. Values have 64 dimensions. Senders that sent the same symbol have the same key and the same deployed value, so the deployed utilities do not depend on how equal scores are ordered. The residual head maps $[h_i,c_i^b,h_i\odot c_i^b]$ through $192\to64\to|\mathcal A|_{\max}$ with ReLU; its last layer is zero-initialized, illegal actions are masked after the head, and an empty mailbox gives a zero correction. The mixer is a state-conditioned QMIX network with embedding width 32 and hypernetwork width 64~\cite{r14}.

\textbf{Training.} After each collected episode, one update is made on 32 episodes sampled from a 5{,}000-episode replay buffer, starting as soon as 32 episodes are stored; the TD and alignment losses are active from the first update. Complete episodes are unrolled to recompute recurrent states. Rewards are standardized with running statistics of the valid steps of the sampled batches (updated before the targets are computed). TD targets use Double-Q learning~\cite{r37}: online local utilities select greedy legal next actions and target networks evaluate them; targets are synchronized every 200 updates. Exploration decays from 1 to 0.05 over 50{,}000 environment steps, the discount is 0.99, Adam learning rates are $5\times10^{-4}$ for the recurrent network, action head and mixer and $3\times10^{-4}$ for codec and receiver, and gradients are clipped at norm 10. Each update draws at most eight probe times from the batch, excluding padding and receivers with at most one legal action; teammate utilities are fixed at $u_k=Q_k^{\refb}(a_k^{\rm replay})$. The same hyperparameters are used on every task. Budgets are 2.05M environment steps (1.2M on navigation), with 100 monitoring episodes every 50{,}000 steps and 1{,}000 final evaluation episodes. The same-backbone QMIX control uses the identical backbone, reward standardization, optimizer, replay, exploration schedule, budget and per-seed initialization, without codec, receiver, reference branch and $\Lcv$.

\section{Tasks, Baselines and Protocol}\label{app:setup}
\begin{table*}[t]
\centering\small
\caption{Tasks. ``Edges'' are the directed message edges that every method may use; all methods use the same one-step delay and the same execution-time information.}
\label{stab:tasks}
\begin{tabularx}{\textwidth}{@{}>{\raggedright\arraybackslash}p{3.2cm}c>{\raggedright\arraybackslash}X>{\raggedright\arraybackslash}p{2.5cm}>{\raggedright\arraybackslash}p{1.8cm}>{\raggedright\arraybackslash}p{1.5cm}@{}}
\toprule
Task & Agents & Description & Edges & Metric & RAVEN \\
\midrule
N1-C2 reference navigation & 2 & MPE \texttt{simple\_reference}: each agent knows only its partner's goal color; 25 steps & pair, both ways & team return & offline, online \\
N1-C3/C4/C6 ring navigation & 3/4/6 & each agent's goal is known only to its predecessor & predecessor $\to$ successor & team return & offline, online \\
N2-L1/L2/L4/L6 broadcast & $1{+}L$ & MPE \texttt{simple\_speaker\_listener}: an immobile speaker sees the target; listeners see landmarks & speaker $\to$ listeners & listener return & offline, online \\
Predator--prey PP1 & 3 & $5\times5$ grid, vision radius 1; all agents must reach the same target within 20 steps & all-to-all & capture success & online \\
SMAC 3m, 8m, MMM, MMM2, 3s5z & 3/8/10/10/8 & StarCraft II micromanagement~\cite{r39}; MMM2 super hard, 3s5z hard & all-to-all, receiver keeps 2 & win rate & online \\
MPE Spread, Tag, Crypto & 3/4/3 & coverage, pursuit and encrypted signaling; in Crypto the eavesdropper Eve is controlled like the other agents and her reward is included in the summed team return; 25 steps & all-to-all & team return & online \\
\bottomrule
\end{tabularx}
\end{table*}

\textbf{Navigation baselines.} SLIM~\cite{r8}, NDQ~\cite{r9}, CACOM~\cite{r5} and ExpoComm~\cite{r6} are run from their authors' official code with their own architectures, losses, optimizers and default hyperparameters---no task-specific tuning is performed on our tasks; only the message edges, the delay and the budget are adapted to the common protocol. MACC~\cite{r20} is implemented from its paper. Every model is trained for 1.2M environment steps. SLIM and MACC use the same four-symbol channel as RAVEN; NDQ sends three float32 values (96 bits), and ExpoComm 64 float32 values (2{,}048 bits). CACOM's two-round exchange is computed inside step $t$, as in the authors' implementation---a 4-value request (8 bits) and a gated 12-value reply at 2 bits per value (24 bits)---and the reply is delivered at $t{+}1$, so every message respects the same one-step delay. In the broadcast settings, RAVEN (offline) and SLIM deploy the one-listener model unchanged; all other methods are retrained at every size, and CACOM and ExpoComm are additionally transferred zero-shot (Table~\ref{stab:zero-shot}).

\textbf{SMAC and MPE reference methods.} The 13 reference methods are QMIX~\cite{r14}, QMIX with larger networks and QMIX with attention~\cite{r48}, QPLEX~\cite{r40}, OW-QMIX and CW-QMIX~\cite{r46}, Soft-QMIX~\cite{r49}, RODE~\cite{r47}, ROCO~\cite{r50}, and the communication methods TeamComm~\cite{r41}, TGCNet~\cite{r42}, CACOM~\cite{r5} and a single-neighbor ExpoComm~\cite{r6} on a QMIX backbone. All of them are run by us from the authors' official open-source implementations under the same 2.05M-step budget; no numbers are taken from prior publications. Their final value is the last evaluation before 2.05M steps, averaged over their runs.

\textbf{Statistics.} The training seed is the statistical unit. Navigation methods and controls use five seeds, the mechanism ablation ten, and the rate study eight; on SMAC, online RAVEN uses four seeds and on MPE five. Every final model is evaluated on 1{,}000 held-out episodes whose environment instances are disjoint from training and teacher selection, always at the final checkpoint. We report means with 95\% $t$-intervals over seeds and paired two-sided $t$-tests on seed-level means, with Holm correction inside each family of comparisons (for example, offline RAVEN against NDQ, CACOM and ExpoComm over the ring settings, or the online implementation against one external method over the eight settings). The evaluation episodes reduce the noise of each model and are never treated as independent samples. The ring comparison with SLIM is repeated on fresh seeds that took no part in method development.

\FloatBarrier
\section{Navigation: Complete Results}\label{app:nav}
Table~\ref{stab:nav-main} extends Table~\ref{tab:nav} with MACC, and Figure~\ref{sfig:nav-endpoints} shows every seed. Table~\ref{stab:nav-tests} lists all paired comparisons, including online against offline RAVEN. Figure~\ref{sfig:nav-curves} and Table~\ref{stab:nav-auc} show the learning curves of the methods that record intermediate evaluations: \on{} has the highest time-averaged return in four of the six settings, and \off{}, which is not trained by trial and error, lies above every curve except on N1-C6. Table~\ref{stab:ring} compares the development and fresh seeds of the ring study, Table~\ref{stab:zero-shot} and Figure~\ref{sfig:zero-shot} the zero-shot broadcast transfer, and Figure~\ref{sfig:lead} the margin over the strongest external method.

\begin{table*}[t]
\centering\small
\caption{Navigation: team return at the fixed final checkpoint (mean$_{\pm\text{s.d.}}$ over 5 training seeds, 1{,}000 held-out episodes per model; higher is better). Best per row in bold, second best underlined. $^{\ddagger}$ RAVEN (offline) and SLIM deploy the L1 model zero-shot on L2/L4/L6; all other methods are retrained at every size. MACC is implemented from its paper; under the enforced one-step delay its communication critic assigns credit to a message within the step it is sent, this attribution no longer matches the action the message influenced, and training diverges, which is why Table~\ref{tab:nav} omits it.}
\label{stab:nav-main}
\setlength{\tabcolsep}{4.2pt}
\begin{tabular}{@{}lccccccc@{}}
\toprule
Setting & \textbf{RAVEN (offline)} & \textbf{RAVEN (online)} & SLIM & NDQ & CACOM & ExpoComm & MACC \\
\midrule
N1-C2 & \textbf{$-$9.59$_{\pm0.60}$} & \underline{$-$12.35$_{\pm1.68}$} & $-$18.81$_{\pm0.76}$ & $-$21.34$_{\pm1.87}$ & $-$17.25$_{\pm0.97}$ & $-$14.98$_{\pm1.44}$ & $-$48.12$_{\pm7.75}$ \\
N1-C3 & \textbf{$-$11.47$_{\pm0.39}$} & \underline{$-$15.21$_{\pm2.78}$} & $-$18.79$_{\pm0.71}$ & $-$19.98$_{\pm0.95}$ & $-$17.84$_{\pm0.15}$ & $-$17.35$_{\pm0.92}$ & $-$43.44$_{\pm15.20}$ \\
N1-C4 & \textbf{$-$13.74$_{\pm0.64}$} & \underline{$-$17.86$_{\pm0.30}$} & $-$20.02$_{\pm0.85}$ & $-$20.06$_{\pm0.67}$ & $-$17.89$_{\pm0.44}$ & $-$17.92$_{\pm0.30}$ & $-$49.47$_{\pm6.68}$ \\
N1-C6 & $-$18.15$_{\pm0.12}$ & $-$17.96$_{\pm0.23}$ & $-$20.13$_{\pm0.72}$ & $-$24.32$_{\pm3.55}$ & \textbf{$-$17.78$_{\pm0.28}$} & \underline{$-$17.94$_{\pm0.26}$} & $-$37.97$_{\pm6.34}$ \\
\midrule
N2-L1 & \textbf{$-$7.70$_{\pm0.17}$} & \underline{$-$7.86$_{\pm0.24}$} & $-$11.89$_{\pm3.01}$ & $-$14.19$_{\pm3.02}$ & $-$12.64$_{\pm4.32}$ & $-$15.42$_{\pm0.28}$ & $-$129.72$_{\pm58.17}$ \\
N2-L2 & \textbf{$-$7.86$_{\pm0.14}$$^{\ddagger}$} & \underline{$-$9.13$_{\pm1.87}$} & $-$12.10$_{\pm2.92}$$^{\ddagger}$ & $-$15.08$_{\pm2.78}$ & $-$9.65$_{\pm3.58}$ & $-$14.16$_{\pm2.39}$ & $-$115.26$_{\pm71.17}$ \\
N2-L4 & \textbf{$-$7.92$_{\pm0.21}$$^{\ddagger}$} & \underline{$-$9.82$_{\pm2.31}$} & $-$12.19$_{\pm2.83}$$^{\ddagger}$ & $-$20.52$_{\pm1.21}$ & $-$10.68$_{\pm3.61}$ & $-$13.66$_{\pm3.09}$ & $-$138.76$_{\pm61.12}$ \\
N2-L6 & \textbf{$-$7.82$_{\pm0.09}$$^{\ddagger}$} & \underline{$-$9.23$_{\pm0.53}$} & $-$12.26$_{\pm2.99}$$^{\ddagger}$ & $-$24.60$_{\pm5.51}$ & $-$13.81$_{\pm5.28}$ & $-$13.98$_{\pm2.40}$ & $-$165.82$_{\pm6.16}$ \\
\bottomrule
\end{tabular}
\end{table*}

\begin{table*}[t]
\centering\small
\caption{Paired seed-level comparisons on navigation: mean difference in team return, Holm-adjusted significance within each family of comparisons ($^{***}p<0.001$, $^{**}p<0.01$, $^{*}p<0.05$, $^{\dagger}p<0.1$) and the number of seeds with positive/negative difference. C3/C4/C6 comparisons with SLIM use fresh seeds that never took part in method development; L2/L4/L6 comparisons with SLIM compare two zero-shot transfers.}
\label{stab:nav-tests}
\setlength{\tabcolsep}{4pt}
\begin{tabular}{@{}lcccccc@{}}
\toprule
\multicolumn{6}{@{}l}{\textit{RAVEN (offline) minus \ldots}} \\
Setting & SLIM & NDQ & CACOM & ExpoComm & MACC \\
\midrule
N1-C2 & $+$9.22$^{***}$ {\scriptsize(5/0)} & $+$11.75$^{***}$ {\scriptsize(5/0)} & $+$7.66$^{***}$ {\scriptsize(5/0)} & $+$5.39$^{***}$ {\scriptsize(5/0)} & $+$38.54$^{**}$ {\scriptsize(5/0)} \\
N1-C3 & $+$7.32$^{***}$ {\scriptsize(5/0)} & $+$8.51$^{***}$ {\scriptsize(5/0)} & $+$6.37$^{***}$ {\scriptsize(5/0)} & $+$5.88$^{***}$ {\scriptsize(5/0)} & $+$31.97$^{*}$ {\scriptsize(5/0)} \\
N1-C4 & $+$6.28$^{***}$ {\scriptsize(5/0)} & $+$6.32$^{***}$ {\scriptsize(5/0)} & $+$4.15$^{***}$ {\scriptsize(5/0)} & $+$4.18$^{***}$ {\scriptsize(5/0)} & $+$35.73$^{**}$ {\scriptsize(5/0)} \\
N1-C6 & $+$1.97$^{**}$ {\scriptsize(5/0)} & $+$6.17$^{\dagger}$ {\scriptsize(5/0)} & $-$0.38$^{\dagger}$ {\scriptsize(1/4)} & $-$0.21 {\scriptsize(1/4)} & $+$19.82$^{*}$ {\scriptsize(5/0)} \\
N2-L1 & $+$4.19$^{\dagger}$ {\scriptsize(4/1)} & $+$6.49$^{**}$ {\scriptsize(5/0)} & $+$4.94$^{\dagger}$ {\scriptsize(4/1)} & $+$7.72$^{***}$ {\scriptsize(5/0)} & $+$122.02$^{*}$ {\scriptsize(5/0)} \\
N2-L2 & $+$4.24$^{\dagger}$ {\scriptsize(4/1)} & $+$7.22$^{*}$ {\scriptsize(5/0)} & $+$1.79 {\scriptsize(2/3)} & $+$6.31$^{*}$ {\scriptsize(5/0)} & $+$107.40$^{*}$ {\scriptsize(5/0)} \\
N2-L4 & $+$4.27$^{\dagger}$ {\scriptsize(4/1)} & $+$12.60$^{***}$ {\scriptsize(5/0)} & $+$2.77 {\scriptsize(5/0)} & $+$5.74$^{\dagger}$ {\scriptsize(5/0)} & $+$130.84$^{*}$ {\scriptsize(5/0)} \\
N2-L6 & $+$4.44$^{\dagger}$ {\scriptsize(4/1)} & $+$16.78$^{*}$ {\scriptsize(5/0)} & $+$5.99 {\scriptsize(4/1)} & $+$6.16$^{*}$ {\scriptsize(5/0)} & $+$158.00$^{***}$ {\scriptsize(5/0)} \\
\midrule\midrule
\multicolumn{7}{@{}l}{\textit{RAVEN (online) minus \ldots}} \\
Setting & SLIM & NDQ & CACOM & ExpoComm & MACC & RAVEN (offline) \\
\midrule
N1-C2 & -- & $+$8.99$^{*}$ {\scriptsize(5/0)} & $+$4.90$^{*}$ {\scriptsize(5/0)} & $+$2.63 {\scriptsize(4/1)} & $+$35.77$^{**}$ {\scriptsize(5/0)} & $-$2.76 {\scriptsize(0/5)} \\
N1-C3 & -- & $+$4.77$^{*}$ {\scriptsize(5/0)} & $+$2.63 {\scriptsize(4/1)} & $+$2.14 {\scriptsize(4/1)} & $+$28.23$^{*}$ {\scriptsize(5/0)} & $-$3.74 {\scriptsize(1/4)} \\
N1-C4 & -- & $+$2.20$^{**}$ {\scriptsize(5/0)} & $+$0.02 {\scriptsize(2/3)} & $+$0.05 {\scriptsize(4/1)} & $+$31.61$^{**}$ {\scriptsize(5/0)} & $-$4.12$^{***}$ {\scriptsize(0/5)} \\
N1-C6 & -- & $+$6.36$^{*}$ {\scriptsize(5/0)} & $-$0.19 {\scriptsize(2/3)} & $-$0.02 {\scriptsize(1/4)} & $+$20.01$^{*}$ {\scriptsize(5/0)} & $+$0.19 {\scriptsize(5/0)} \\
N2-L1 & -- & $+$6.32$^{*}$ {\scriptsize(5/0)} & $+$4.78 {\scriptsize(5/0)} & $+$7.56$^{***}$ {\scriptsize(5/0)} & $+$121.85$^{*}$ {\scriptsize(5/0)} & $-$0.16 {\scriptsize(2/3)} \\
N2-L2 & -- & $+$5.95$^{*}$ {\scriptsize(5/0)} & $+$0.52 {\scriptsize(2/3)} & $+$5.03 {\scriptsize(5/0)} & $+$106.13$^{*}$ {\scriptsize(5/0)} & $-$1.27 {\scriptsize(0/5)} \\
N2-L4 & -- & $+$10.70$^{*}$ {\scriptsize(5/0)} & $+$0.86 {\scriptsize(3/2)} & $+$3.84 {\scriptsize(4/1)} & $+$128.94$^{*}$ {\scriptsize(5/0)} & $-$1.90 {\scriptsize(0/5)} \\
N2-L6 & -- & $+$15.37$^{*}$ {\scriptsize(5/0)} & $+$4.57 {\scriptsize(3/2)} & $+$4.75$^{\dagger}$ {\scriptsize(5/0)} & $+$156.59$^{***}$ {\scriptsize(5/0)} & $-$1.41$^{*}$ {\scriptsize(0/5)} \\
\bottomrule
\end{tabular}
\end{table*}

\begin{figure*}[t]
\centering
\includegraphics[width=\textwidth]{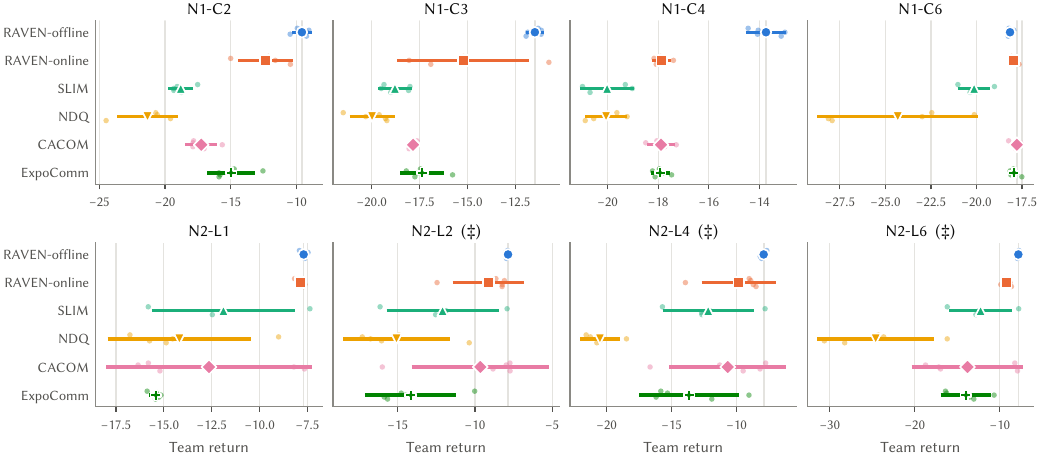}
\caption{Final team return of every training seed (small dots), the seed mean (large marker) and its 95\% interval (bar) in the eight navigation settings. MACC is omitted for scale.}
\label{sfig:nav-endpoints}
\end{figure*}

\begin{figure*}[t]
\centering
\includegraphics[width=\textwidth]{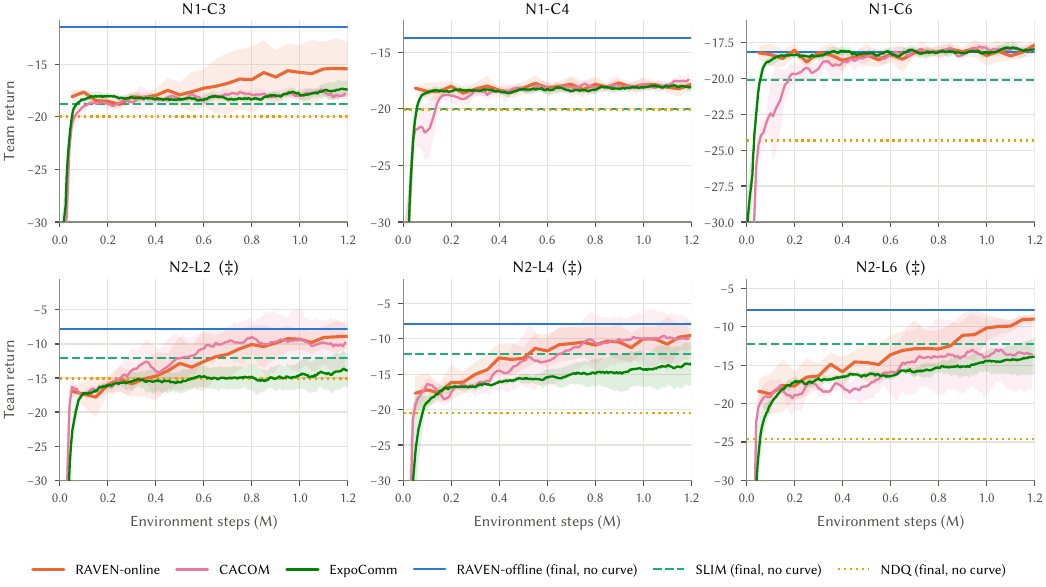}
\caption{Navigation learning curves (mean $\pm$ s.d.\ over 5 seeds). Horizontal lines are final values of methods without recorded curves; the y-axis is clipped at $-30$.}
\label{sfig:nav-curves}
\end{figure*}

\begin{table}[t]
\centering\small
\caption{Sample efficiency on navigation: time-averaged team return over 0.05--1.2M training steps (area under the mean learning curve divided by its length; higher is better). Only methods that record intermediate evaluations are listed.}
\label{stab:nav-auc}
\begin{tabular}{@{}lcccc@{}}
\toprule
Setting & RAVEN (online) & CACOM & ExpoComm & MACC \\
\midrule
N1-C3 & \textbf{$-$17.03} & $-$18.19 & $-$18.10 & $-$46.92 \\
N1-C4 & \textbf{$-$18.12} & $-$18.47 & $-$18.29 & $-$47.44 \\
N1-C6 & $-$18.32 & $-$18.87 & \textbf{$-$18.26} & $-$38.22 \\
N2-L2 & $-$12.54 & \textbf{$-$11.95} & $-$15.35 & $-$115.92 \\
N2-L4 & \textbf{$-$12.48} & $-$13.01 & $-$15.68 & $-$135.90 \\
N2-L6 & \textbf{$-$13.69} & $-$15.88 & $-$16.18 & $-$148.49 \\
\bottomrule
\end{tabular}
\end{table}

\begin{table}[t]
\centering\small
\caption{Ring navigation, RAVEN (offline) minus SLIM on the development seeds and on fresh confirmation seeds, and the return each method loses when its incoming messages are removed at test time (message dependence, fresh seeds).}
\label{stab:ring}
\setlength{\tabcolsep}{3.2pt}
\begin{tabular}{@{}lcccccc@{}}
\toprule
 & \multicolumn{2}{c}{development seeds} & \multicolumn{2}{c}{fresh seeds} & \multicolumn{2}{c}{gain from messages} \\
\cmidrule(lr){2-3}\cmidrule(lr){4-5}\cmidrule(lr){6-7}
Setting & $\Delta$ & seeds $>0$ & $\Delta$ & seeds $>0$ & RAVEN & SLIM \\
\midrule
N1-C3 & $+$7.04 & 5/5 & $+$7.32 & 5/5 & $+$6.44 & $+$1.17 \\
N1-C4 & $+$5.52 & 5/5 & $+$6.28 & 5/5 & $+$4.64 & $+$0.49 \\
N1-C6 & $+$2.49 & 5/5 & $+$1.97 & 5/5 & $+$0.41 & $+$0.28 \\
\bottomrule
\end{tabular}
\end{table}

\begin{table}[t]
\centering\small
\caption{Zero-shot scale transfer on the broadcast family: a model trained with one listener is deployed unchanged to 2/4/6 listeners (``zero-shot'') and compared with the same method retrained at each size.}
\label{stab:zero-shot}
\footnotesize\setlength{\tabcolsep}{2.6pt}
\begin{tabular}{@{}lcccccc@{}}
\toprule
 & RAVEN (offline) & SLIM & \multicolumn{2}{c}{CACOM} & \multicolumn{2}{c}{ExpoComm} \\
\cmidrule(lr){2-2}\cmidrule(lr){3-3}\cmidrule(lr){4-5}\cmidrule(lr){6-7}
Setting & zero-shot & zero-shot & zero-shot & retrained & zero-shot & retrained \\
\midrule
N2-L2 & \textbf{$-$7.86} & $-$12.10 & $-$12.96 & $-$9.65 & $-$15.78 & $-$14.16 \\
N2-L4 & \textbf{$-$7.92} & $-$12.19 & $-$12.94 & $-$10.68 & $-$15.75 & $-$13.66 \\
N2-L6 & \textbf{$-$7.82} & $-$12.26 & $-$12.95 & $-$13.81 & $-$15.78 & $-$13.98 \\
\bottomrule
\end{tabular}
\end{table}

\begin{figure}[t]
\centering
\includegraphics[width=\columnwidth]{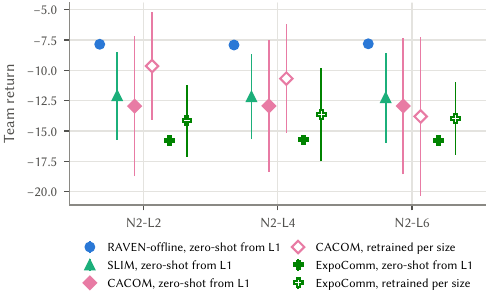}
\caption{Zero-shot transfer from one listener (filled markers) against retraining at every size (hollow markers); mean and 95\% interval over 5 seeds.}
\label{sfig:zero-shot}
\end{figure}

\begin{figure}[t]
\centering
\includegraphics[width=\columnwidth]{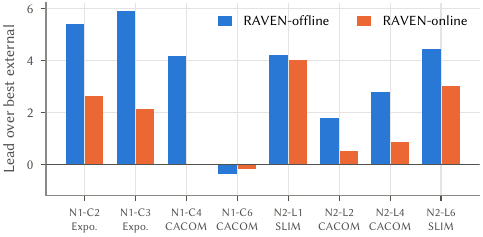}
\caption{Margin of each RAVEN implementation over the strongest external method of each setting (named under the axis).}
\label{sfig:lead}
\end{figure}

\section{Mechanism: Complete Results}\label{app:mechanism}
Table~\ref{stab:ablation} gives all ablation variants with $K=4$ and $K=2$. Table~\ref{stab:mechanism} reports the agreement between the diagnostics of Proposition~\ref{prop:selection} and return for all four methods and for the three condition variants alone. Table~\ref{stab:rate} and Figure~\ref{sfig:rate} report the frozen offline models under reduced sending rates and message loss: at one fresh symbol every four steps, RAVEN still beats the every-step learned channel by $+5.26$ and the unweighted condition at the same rate by $+4.45$ (8/8 seeds each, Holm-significant). Figure~\ref{sfig:dependence} shows the message dependence and the symbol entropy (1.94 bits on average of 2 with $K=4$; 0.93 of 1 with $K=2$).

\begin{table*}[t]
\centering\small
\caption{Mechanism ablation of RAVEN (offline) on N1-C2 (10 independently trained seeds, each with its own teacher; all variants share the teacher and construction data of their seed and match capacity and number of updates). $\Delta$: full method minus variant (paired, 95\% CI). $p$: Holm-adjusted within the pre-registered K=4 family (the sender-freezing comparison belongs to the candidate-internal family). K=2 comparisons are descriptive ($^{\circ}$: computed from the per-seed records, not part of a registered family).}
\label{stab:ablation}
\setlength{\tabcolsep}{4.2pt}
\begin{tabular}{@{}llcccccc@{}}
\toprule
 & & \multicolumn{4}{c}{K = 4 (2 bits)} & \multicolumn{2}{c}{K = 2 (1 bit)} \\
\cmidrule(lr){3-6}\cmidrule(lr){7-8}
Variant & changes & return & $\Delta$ & seeds $+$/$-$ & $p$ (Holm) & return & $\Delta$ \\
\midrule
\textbf{RAVEN (offline), full} & -- & $-$9.45$_{\pm0.46}$ & -- & -- & -- & $-$14.78 & -- \\
local-residual target + geometric condition & target and condition & $-$10.75$_{\pm1.72}$ & $+$1.30 [0.10, 2.50] & 7/2 & 0.146 & $-$15.31 & $+$0.53 \\
geometric receiver condition & condition & $-$13.54$_{\pm1.79}$ & $+$4.09 [2.85, 5.32] & 10/0 & $2.6\times10^{-4}$ & $-$16.12 & $+$1.33 \\
plain learnable 4-symbol channel (no codebook) & codebook construction & $-$14.49$_{\pm0.65}$ & $+$5.04 [4.62, 5.46] & 10/0 & $9.7\times10^{-9}$ & $-$16.80 & $+$2.02 \\
sender unfrozen after distillation & sender freezing & $-$14.57$_{\pm0.99}$ & $+$5.12 [4.53, 5.72] & 10/0 & $3.5\times10^{-8}$ & $-$17.37 & $+$2.59$^{\circ}$ \\
joint-Q sensitivity condition & condition & $-$15.21$_{\pm1.46}$ & $+$5.76 [4.76, 6.77] & 10/0 & $4.8\times10^{-6}$ & $-$17.79 & $+$3.01 \\
full local observation condition & condition & $-$16.99$_{\pm0.39}$ & $+$7.54 [7.18, 7.91] & 10/0 & $8.1\times10^{-11}$ & $-$18.69 & $+$3.91 \\
no receiver condition (marginal average) & condition & $-$17.63$_{\pm0.41}$ & $+$8.18 [7.76, 8.60] & 10/0 & $1.3\times10^{-10}$ & $-$18.64 & $+$3.86 \\
no messages (retrained) & communication & $-$19.27$_{\pm0.19}$ & $+$9.82 [9.44, 10.21] & 10/0 & $1.4\times10^{-11}$ & $-$19.27 & $+$4.49 \\
\bottomrule
\end{tabular}
\end{table*}

\begin{table}[t]
\centering\small
\caption{Does the finite objective predict task return? For each seed, diagnostics are computed on held-out samples for RAVEN (offline) and three condition variants (``all four'') or for the three variants only; entries are the mean over 10 seeds of pairwise direction concordance / Spearman $\rho$ between lower diagnostic and higher return.}
\label{stab:mechanism}
\setlength{\tabcolsep}{3pt}
\begin{tabular}{@{}lcccc@{}}
\toprule
 & \multicolumn{2}{c}{all four methods} & \multicolumn{2}{c}{controls only} \\
\cmidrule(lr){2-3}\cmidrule(lr){4-5}
Diagnostic & K=4 & K=2 & K=4 & K=2 \\
\midrule
total risk $R$ & 0.833 / 0.76 & 0.800 / 0.72 & 0.667 / 0.40 & 0.633 / 0.35 \\
abstraction term $A$ & 0.850 / 0.78 & 0.817 / 0.72 & 0.700 / 0.45 & 0.667 / 0.35 \\
compression term $D$ & 0.800 / 0.72 & 0.617 / 0.26 & 0.867 / 0.80 & 0.833 / 0.70 \\
teacher-action regret & 0.783 / 0.72 & 0.767 / 0.66 & 0.767 / 0.65 & 0.767 / 0.60 \\
\bottomrule
\end{tabular}
\end{table}

\begin{table*}[t]
\centering\small
\caption{Communication rate and loss on N1-C2 (8 training seeds, frozen models). The receiver keeps the last symbol that arrived; ``every $k$'' sends a fresh symbol every $k$ steps, ``once'' only at the first step, and ``loss'' drops each message independently. Entries are mean team returns.}
\label{stab:rate}
\begin{tabular}{@{}lcccccc@{}}
\toprule
Method & every step & every 2 & every 4 & once & 10\% loss & 30\% loss \\
\midrule
bytes / episode & 50 & 26 & 14 & 2 & 50 & 50 \\
\midrule
RAVEN (offline), full & $-$9.78 & $-$9.78 & $-$9.79 & $-$9.84 & $-$9.80 & $-$9.88 \\
local-residual target + geometric condition & $-$11.60 & $-$11.60 & $-$11.61 & $-$11.71 & $-$11.62 & $-$11.68 \\
geometric receiver condition & $-$14.20 & $-$14.22 & $-$14.24 & $-$14.37 & $-$14.22 & $-$14.25 \\
plain learnable 4-symbol channel (no codebook) & $-$15.05 & $-$15.10 & $-$15.15 & $-$15.73 & $-$15.07 & $-$15.10 \\
\bottomrule
\end{tabular}
\end{table*}

\begin{figure}[t]
\centering
\includegraphics[width=\columnwidth]{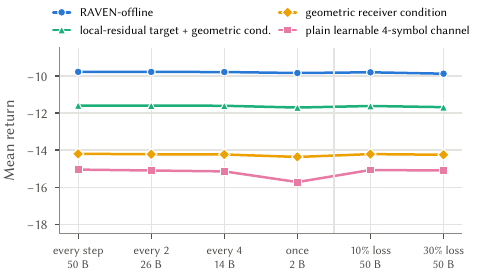}
\caption{Return under reduced sending rates and message loss (N1-C2, 8 seeds). Numbers under the axis give the serialized payload per episode.}
\label{sfig:rate}
\end{figure}

\begin{figure}[t]
\centering
\includegraphics[width=\columnwidth]{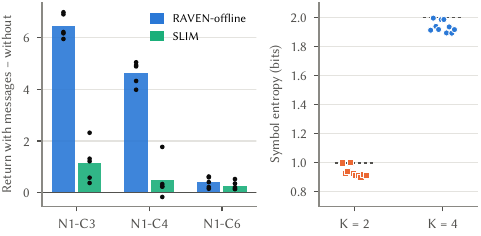}
\caption{Left: return lost when incoming messages are removed at test time (ring navigation, fresh seeds). Right: empirical symbol entropy of offline RAVEN on N1-C2 for each of 10 seeds; dashed lines mark the maximum.}
\label{sfig:dependence}
\end{figure}

\section{Online RAVEN: Complete Results}\label{app:scale}
Table~\ref{stab:controls} lists the controlled comparisons, and Figures~\ref{sfig:pp1} and~\ref{sfig:online-ablation} show them per seed. Figure~\ref{sfig:dynamics} shows the TD losses of both branches and $\Lcv$: the two TD losses track each other, as intended by the shared-weight design, and $\Lcv$ stays small but non-zero, so the codec keeps adapting to the receivers' action values throughout training. Table~\ref{stab:smac-mpe} lists all 14 methods on SMAC and MPE, Table~\ref{stab:smac-speed} the learning speed, and Figures~\ref{sfig:smac-curves}--\ref{sfig:heatmap} the curves and the normalized scores.

\begin{table*}[t]
\centering\small
\caption{Online implementation against controls that differ only in the communication pathway. The same-backbone QMIX shares the GRU backbone, optimizer, replay, exploration schedule, budget and per-seed initialization and removes codec, receiver, reference branch and $\mathcal{L}_{\mathrm{CV}}$; the ablation removes only the reference-branch TD loss. $p$: paired two-sided $t$-test (PP1: Holm-adjusted).}
\label{stab:controls}
\begin{tabular}{@{}lllcccc c@{}}
\toprule
Task & Metric & Control & RAVEN & control & $\Delta$ [95\% CI] & seeds $+$/$-$ & $p$ \\
\midrule
Predator--prey PP1 & capture success (\%) & same-backbone QMIX & \textbf{96.04} & 53.20 & $+$42.84 [15.5, 70.2] & 5/0 & 0.024 \\
MPE Spread & team return & same-backbone QMIX & \textbf{$-$29.00} & $-$31.87 & $+$2.87 [1.77, 3.97] & 5/0 & 0.002 \\
MPE Spread & team return & w/o reference TD loss & \textbf{$-$28.83} & $-$29.52 & $+$0.69 [0.06, 1.32] & 3/0 & 0.043 \\
\bottomrule
\end{tabular}
\end{table*}

\begin{figure}[t]
\centering
\includegraphics[width=\columnwidth]{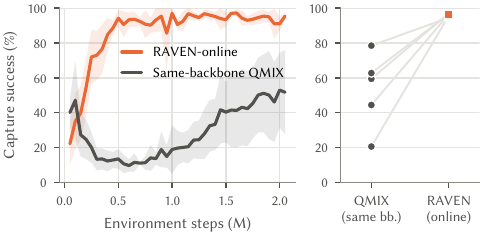}
\caption{Predator--prey. Left: capture success during training (mean $\pm$ s.d., 5 seeds). Right: final success of each seed; lines connect runs with the same seed and initialization.}
\label{sfig:pp1}
\end{figure}

\begin{figure}[t]
\centering
\includegraphics[width=\columnwidth]{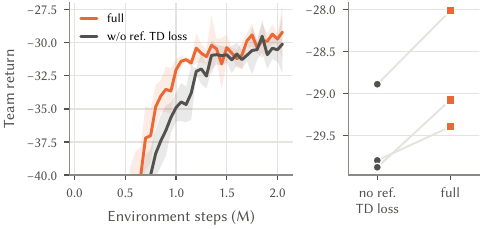}
\caption{Spread: online RAVEN against the variant without the reference-branch TD loss (3 matched seeds). Left: learning curves; right: final return per seed.}
\label{sfig:online-ablation}
\end{figure}

\begin{figure*}[t]
\centering
\includegraphics[width=\textwidth]{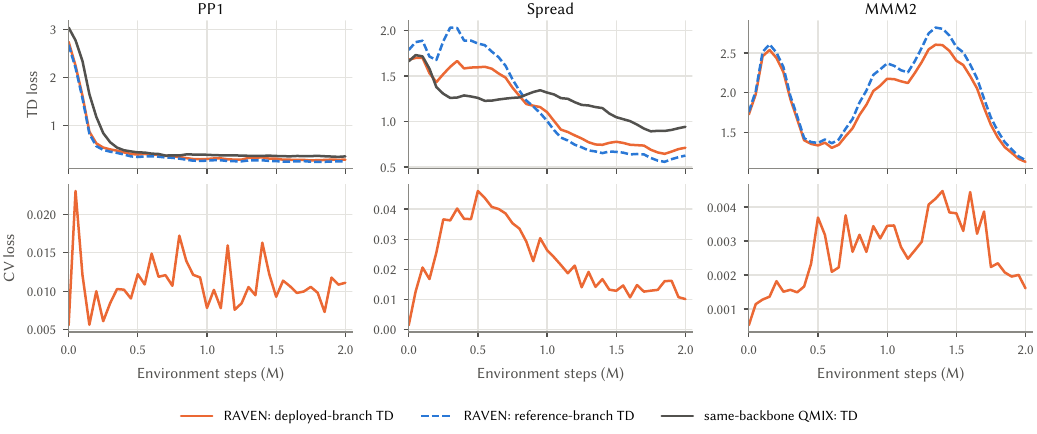}
\caption{Training losses (means over seeds, 50k-step bins). Top: TD loss of the deployed branch, of the reference branch and of the same-backbone QMIX. Bottom: the alignment loss $\Lcv$.}
\label{sfig:dynamics}
\end{figure*}

\begin{table*}[t]
\centering\small
\caption{SMAC (win rate, \%) and MPE (team return) at 2.05M environment steps. RAVEN (online): mean$_{\pm\text{s.d.}}$ (SMAC: 4 seeds; MPE: 5 seeds), 1{,}000 final evaluation episodes per model. Reference methods: mean over our runs (4--9) of the authors' official open-source code, at the last evaluation before 2.05M steps (about 100 test episodes). Best per column in bold, second underlined (ties marked alike). The last two columns give the mean per-task min--max normalized score and the mean rank over the eight tasks. $^{\S}$ single-neighbor ExpoComm variant.}
\label{stab:smac-mpe}
\setlength{\tabcolsep}{3.4pt}
\begin{tabular}{@{}lccccc|ccc|cc@{}}
\toprule
Method & 3m & 8m & MMM & MMM2 & 3s5z & Spread & Tag & Crypto & norm. & rank \\
\midrule
\textbf{RAVEN (online)} & 99.72$_{\pm0.22}$ & \textbf{99.95$_{\pm0.06}$} & \textbf{100.00$_{\pm0.00}$} & \underline{87.35$_{\pm8.30}$} & \textbf{97.80$_{\pm1.12}$} & \underline{$-$29.00$_{\pm0.59}$} & \textbf{277.3$_{\pm3.0}$} & \underline{48.03$_{\pm0.04}$} & 0.98 & 1.75 \\
\midrule
QMIX & \underline{99.80} & 98.80 & 98.60 & 57.80 & 86.40 & $-$43.42 & 23.39 & 0.38 & 0.74 & 7.3 \\
QMIX (large) & 97.57 & 99.31 & 96.88 & 26.91 & 94.91 & $-$44.17 & $-$46.04 & 21.69 & 0.70 & 8.6 \\
QPLEX & 99.75 & \underline{99.50} & 30.75 & 62.00 & 96.50 & $-$31.57 & 235.1 & 45.44 & 0.80 & 5.0 \\
OW-QMIX & \textbf{100.00} & 96.20 & 98.00 & 38.75 & 89.60 & $-$51.50 & $-$51.09 & 8.70 & 0.69 & 8.6 \\
CW-QMIX & 98.00 & 96.50 & 96.00 & 2.25 & 80.50 & $-$49.39 & $-$22.87 & 47.21 & 0.69 & 9.9 \\
RODE & 96.75 & 98.50 & 98.00 & 14.00 & 63.00 & $-$60.18 & 106.9 & 7.25 & 0.66 & 9.9 \\
ROCO & 74.50 & 69.25 & 97.25 & 30.25 & 94.00 & $-$77.83 & 95.39 & 21.89 & 0.54 & 9.9 \\
Soft-QMIX & 99.00 & 88.75 & 28.50 & 0.00 & 0.00 & $-$140.0 & $-$58.21 & 36.72 & 0.32 & 11.9 \\
QMIX-att & 97.71 & 98.96 & \underline{99.79} & 42.92 & 91.67 & $-$124.4 & $-$42.83 & \textbf{50.00} & 0.68 & 7.5 \\
QMIX+TeamComm & 99.25 & 99.00 & 98.00 & 64.25 & 91.00 & $-$43.73 & 68.56 & 48.00 & 0.84 & 5.6 \\
QMIX+TGCNet & 99.00 & 99.25 & 74.00 & 66.50 & 96.75 & $-$44.95 & 55.46 & 48.00 & 0.80 & 5.9 \\
CACOM & 96.25 & 43.50 & 98.75 & 34.25 & \underline{97.25} & $-$48.47 & 251.1 & 3.18 & 0.67 & 8.1 \\
ExpoComm (1-nbr.)$^{\S}$ & 99.00 & \underline{99.50} & \textbf{100.00} & \textbf{94.50} & 96.25 & \textbf{$-$26.80} & \underline{276.9} & $-$31.20 & 0.87 & 4.0 \\
\midrule
Same-backbone QMIX & -- & -- & -- & -- & -- & $-$31.87 & -- & -- & -- & -- \\
\bottomrule
\end{tabular}
\end{table*}

\begin{table*}[t]
\centering\small
\caption{SMAC learning speed: time-averaged win rate over 0.05--2.05M steps (\%) / first training step (M) at which the mean curve reaches 80\% (``--'': never within budget). Reference curves are binned to 50k steps.}
\label{stab:smac-speed}
\begin{tabular}{@{}lccccc@{}}
\toprule
Method & 3m & 8m & MMM & MMM2 & 3s5z \\
\midrule
RAVEN (online) & 92.5 / 0.25 & 93.2 / 0.25 & 92.8 / 0.30 & 44.8 / 1.50 & 61.4 / 1.20 \\
QMIX & 90.8 / 0.40 & 93.1 / 0.15 & 86.3 / 0.45 & 29.3 / -- & 59.1 / 1.15 \\
QPLEX & 97.8 / 0.15 & 97.2 / 0.10 & 86.9 / 0.30 & 25.9 / -- & 90.0 / 0.30 \\
QMIX+TeamComm & 82.1 / 0.60 & 91.9 / 0.20 & 88.7 / 0.35 & 34.8 / -- & 64.8 / 1.00 \\
QMIX+TGCNet & 78.6 / 0.60 & 93.1 / 0.30 & 84.0 / 0.50 & 42.5 / -- & 67.3 / 1.00 \\
CACOM & 92.2 / 0.20 & 38.9 / -- & 82.9 / 0.45 & 12.4 / -- & 64.7 / 0.90 \\
ExpoComm (1-nbr.)$^{\S}$ & 92.1 / 0.30 & 91.9 / 0.35 & 87.2 / 0.40 & 58.1 / 1.10 & 74.6 / 0.70 \\
\bottomrule
\end{tabular}
\end{table*}

\begin{figure*}[t]
\centering
\includegraphics[width=\textwidth]{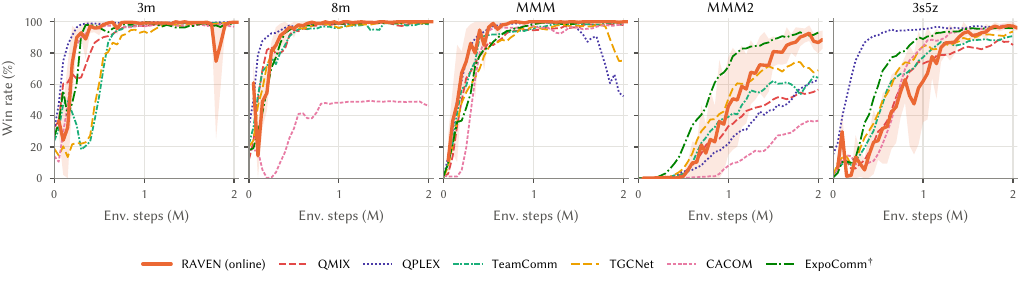}
\caption{SMAC learning curves of online RAVEN (mean $\pm$ s.d.) and six reference methods (mean over runs, binned to 50k steps). $^{\dagger}$Single-neighbor variant.}
\label{sfig:smac-curves}
\end{figure*}

\begin{figure*}[t]
\centering
\includegraphics[width=\textwidth]{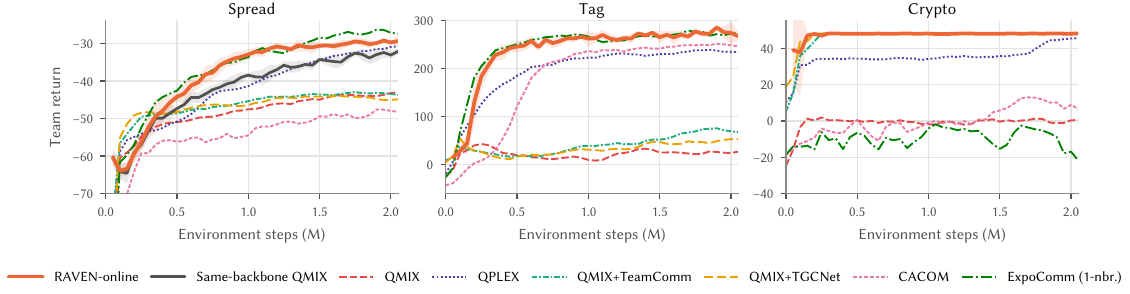}
\caption{MPE learning curves. Online RAVEN and the same-backbone QMIX: mean $\pm$ s.d.\ over 5 seeds; reference methods: mean over runs, binned to 50k steps.}
\label{sfig:mpe-curves}
\end{figure*}

\begin{figure}[t]
\centering
\includegraphics[width=\columnwidth]{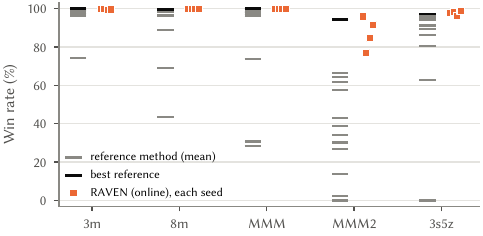}
\caption{Final SMAC win rate of each online RAVEN seed (squares) against the mean of every reference method (ticks; the best one in black).}
\label{sfig:smac-seeds}
\end{figure}

\begin{figure}[t]
\centering
\includegraphics[width=\columnwidth]{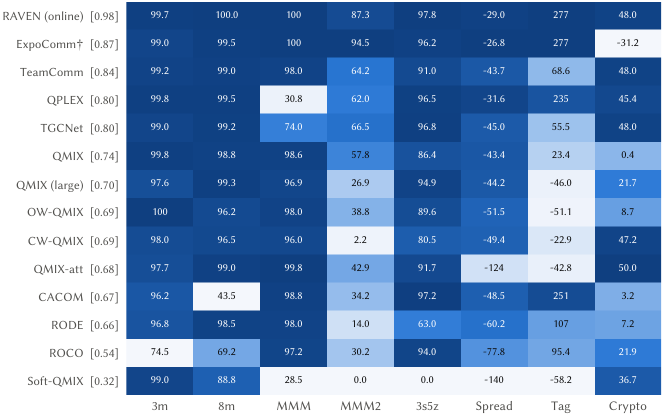}
\caption{Final score of each method on each task, colored by the per-task min--max normalized score (cell values are win rates or returns; the mean normalized score is in brackets).}
\label{sfig:heatmap}
\end{figure}

\FloatBarrier
\section{Cost: Complete Results}\label{app:cost}
Parameters and execution FLOPs (matrix multiplications of one team decision step) are counted on networks rebuilt from the code and configurations of the runs; for both RAVEN implementations, NDQ and MACC the rebuilt counts equal the parameter counts recorded by the runs. Bits are logical payload sizes. Measured traffic is the number of bits delivered during a complete training run. Training time is compared only between runs on the same machine: offline RAVEN (teacher included) and SLIM on one GPU server, and online RAVEN, CACOM and MACC on one Apple M4 Pro CPU. Table~\ref{stab:efficiency} gives all savings, Table~\ref{stab:cost-abs} representative absolute values, and Figures~\ref{sfig:time} and~\ref{sfig:efficiency} the full picture.

\begin{figure}[t]
\centering
\includegraphics[width=\columnwidth]{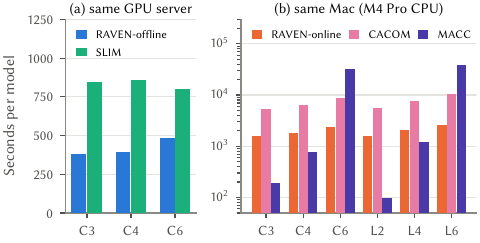}
\caption{Training time per model, compared within one machine: (a) offline RAVEN (teacher included) and SLIM; (b) online RAVEN, CACOM and MACC (log scale).}
\label{sfig:time}
\end{figure}

\begin{table*}[t]
\centering\small
\caption{Resource savings of RAVEN against every external communication-learning method, $100(1-\text{RAVEN}/\text{other})$ in \%; ``lo\,/\,hi'' is the range over settings or maps; negative values mean RAVEN uses more. Parameters and FLOPs (one team decision step) are counted on networks rebuilt from the code and configurations of the runs. Measured traffic is the number of bits delivered during the whole training run. Training time is compared only between runs on the same machine (offline vs.\ SLIM on one GPU server; online vs.\ CACOM and MACC on one Mac). $^{\star}$ lower bound: only the request round of CACOM and only the message vector of TeamComm/TGCNet are counted.}
\label{stab:efficiency}
\begin{tabular}{@{}llccccc@{}}
\toprule
RAVEN & vs. & params & exec.\ FLOPs & bits / message & measured traffic & training time \\
\midrule
RAVEN (offline) & SLIM & 95.6\,/\,98.0 & -- & 0.00 & -- & 39.6\,/\,55.4 \\
RAVEN (offline) & NDQ & $-$1.4\,/\,85.1 & 82.2\,/\,96.8 & 97.92 & -- & -- \\
RAVEN (offline) & CACOM & $-$2.1\,/\,85.0 & 85.8\,/\,97.7 & 91.67 & -- & -- \\
RAVEN (offline) & ExpoComm & 48.4\,/\,92.3 & 91.5\,/\,98.3 & 99.90 & -- & -- \\
RAVEN (offline) & MACC & $-$229.0\,/\,50.9 & 45.4\,/\,89.3 & 0.00 & -- & -- \\
\midrule
RAVEN (online) & SLIM & 81.7\,/\,81.8 & -- & 0.00 & -- & -- \\
RAVEN (online) & NDQ & $-$47.1\,/\,$-$37.2 & $-$48.3\,/\,$-$39.6 & 97.92 & 97.9 & -- \\
RAVEN (online) & CACOM & $-$48.3\,/\,$-$38.2 & $-$18.6\,/\,1.6 & 91.67 & 79.3\,/\,90.0 & 70.6\,/\,74.3 \\
RAVEN (online) & ExpoComm & 29.2\,/\,29.4 & 28.2\,/\,29.0 & 99.90 & 99.8 & -- \\
RAVEN (online) & MACC & $-$351.8\,/\,$-$348.9 & $-$363.5\,/\,$-$355.8 & 0.00 & -- & $-$1525.1\,/\,93.1 \\
\midrule
RAVEN (online), SMAC & CACOM & $-$42.3\,/\,$-$24.2 & $-$4.5\,/\,41.2 & 87.50$^{\star}$ & -- & -- \\
RAVEN (online), SMAC & QMIX+TeamComm & 23.7\,/\,28.2 & 4.7\,/\,19.8 & 99.90$^{\star}$ & -- & -- \\
RAVEN (online), SMAC & QMIX+TGCNet & 33.6\,/\,38.8 & $-$2.9\,/\,12.7 & 99.90$^{\star}$ & -- & -- \\
RAVEN (online), SMAC & ExpoComm (1-nbr.)$^{\S}$ & 78.8\,/\,81.4 & 78.5\,/\,81.4 & 99.95 & -- & -- \\
\midrule
RAVEN (online), MPE & CACOM & $-$54.4\,/\,$-$51.8 & $-$35.6\,/\,$-$28.7 & 87.50$^{\star}$ & -- & -- \\
RAVEN (online), MPE & QMIX+TeamComm & 29.2\,/\,29.6 & 16.3\,/\,21.0 & 99.90$^{\star}$ & -- & -- \\
RAVEN (online), MPE & QMIX+TGCNet & 39.8\,/\,40.2 & 8.4\,/\,13.7 & 99.90$^{\star}$ & -- & -- \\
RAVEN (online), MPE & ExpoComm (1-nbr.)$^{\S}$ & 81.9\,/\,82.1 & 81.8\,/\,82.1 & 99.95 & -- & -- \\
\bottomrule
\end{tabular}
\end{table*}

\begin{table*}[t]
\centering\footnotesize
\caption{Absolute cost. Top: navigation, deployed parameters / execution FLOPs per team decision step (SLIM: parameters recorded by its runs, no FLOP count). Bottom: SMAC and MPE, parameters / FLOPs / bits per agent message (lower bounds for CACOM, TeamComm and TGCNet).}
\label{stab:cost-abs}
\setlength{\tabcolsep}{3.5pt}
\begin{tabular}{@{}lccccccc@{}}
\toprule
Setting & RAVEN (offline) & RAVEN (online) & SLIM & NDQ & CACOM & ExpoComm & MACC \\
\midrule
N1-C2 & 11.7k / 23k & 48.0k / 190k & 264.0k / -- & 32.7k / 128k & 32.4k / 160k & 68.0k / 268k & 10.7k / 42k \\
N1-C6 & 35.2k / 68k & 48.3k / 578k & -- / -- & 34.7k / 409k & 34.5k / 567k & 68.2k / 807k & 10.7k / 125k \\
N2-L1 & 5.3k / 10k & 48.0k / 190k & 264.0k / -- & 32.7k / 128k & 32.4k / 160k & 68.0k / 268k & 10.7k / 42k \\
N2-L6 & 5.3k / 16k & 48.3k / 677k & 264.0k / -- & 35.2k / 485k & 35.0k / 688k & 68.3k / 943k & 10.7k / 146k \\
\bottomrule
\end{tabular}

\vspace{4pt}
\begin{tabular}{@{}lccccc@{}}
\toprule
Setting & RAVEN (online) & CACOM & QMIX+TeamComm & QMIX+TGCNet & ExpoComm (1-nbr.)$^{\S}$ \\
\midrule
3m & 50.1k / 0.30M / 2 & 35.2k / 0.29M / 16 & 69.8k / 0.37M / 2048 & 81.9k / 0.34M / 2048 & 269.4k / 1.61M / 4096 \\
MMM2 & 61.6k / 1.24M / 2 & 48.8k / 2.10M / 16 & 80.7k / 1.30M / 2048 & 92.8k / 1.20M / 2048 & 290.2k / 5.77M / 4096 \\
Spread & 48.5k / 0.29M / 2 & 31.4k / 0.21M / 16 & 68.5k / 0.36M / 2048 & 80.6k / 0.33M / 2048 & 267.4k / 1.59M / 4096 \\
\bottomrule
\end{tabular}
\end{table*}

\begin{figure*}[t]
\centering
\includegraphics[width=\textwidth]{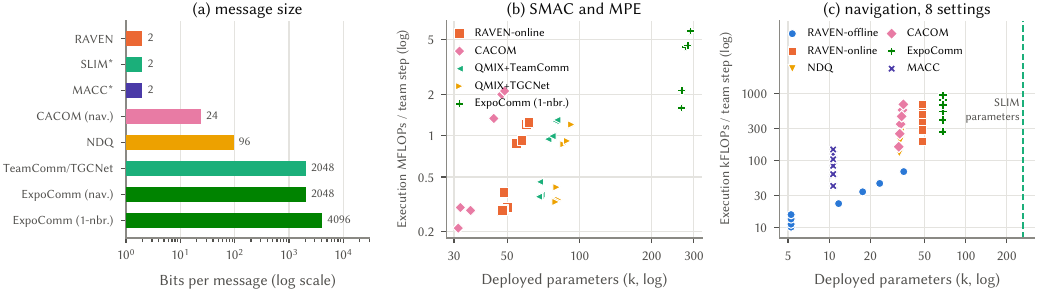}
\caption{(a) Bits per message. (b) Deployed parameters against execution FLOPs per team decision step on SMAC and MPE. (c) The same for the eight navigation settings; SLIM's recorded parameter count is marked by the dashed line.}
\label{sfig:efficiency}
\end{figure*}

\FloatBarrier
\end{document}